\documentclass[10pt]{article}
\usepackage{fullpage,graphicx,subfigure,mathdots,mathpazo,color}
\usepackage{amsmath,amscd,tikz,mathrsfs,cite}
\usepackage[normalem]{ulem}
\usepackage{amsmath}
\usepackage{setspace}

\usepackage{epsfig,amsmath,graphicx,amssymb,overpic}
\usepackage{bm}
\usepackage{graphicx}
\usepackage{subfigure, xcolor}
\usepackage{ntheorem}
\usepackage{listings,diagbox}

\usepackage{color}
\usepackage{epsfig,amsmath,graphicx,amssymb,overpic}

\usepackage{booktabs}
\usepackage{diagbox}
\usepackage{graphicx}
\usepackage{dcolumn}
\usepackage{bm}
\usepackage{graphicx}
\usepackage{subfigure}
\usepackage{epsfig,amsmath,graphicx,amssymb,overpic,cite}
\usepackage{makecell}

\usepackage{epsfig,amsmath,graphicx,amssymb,overpic}
\usepackage{bm}
\usepackage{graphicx}
\usepackage{subfigure, xcolor}
\usepackage{float}
\usepackage{listings}
\usepackage{tabularx}
\usepackage{multirow}

\usepackage{algorithm}
\usepackage{algorithmic}

\def\be{\begin{equation}}
\def\ee{\end{equation}}
\def\bee{\begin{eqnarray}}
\def\ene{\end{eqnarray}}
\def\bes{\begin{subequations}}
\def\ees{\end{subequations}}

\newcommand{\PT}{{\cal PT}}
\newcommand{\dd}{\mathrm{d}}

\def\v{\vspace{0.1in}}

\def\be{\begin{equation}}
\def\ee{\end{equation}}
\def\bee{\begin{eqnarray}}
\def\ene{\end{eqnarray}}
\def\bes{\begin{subequations}}
\def\ees{\end{subequations}}

\def\d{\displaystyle}

\def\v{\vspace{0.1in}}

\newtheorem{lemma}{Lemma}
\newtheorem{theorem}{Theorem}
\newtheorem{definition}{Definition}
\newtheorem*{proof}{Proof}
\newtheorem{remark}{Remark}

\usepackage{geometry}
\allowdisplaybreaks[4]

\begin{document}

\baselineskip=14pt
\renewcommand {\thefootnote}{\dag}
\renewcommand {\thefootnote}{\ddag}
\renewcommand {\thefootnote}{ }

\pagestyle{plain}

\begin{center}
\baselineskip=16pt \leftline{} \vspace{-.3in} {\Large \bf Two-stage initial-value iterative physics-informed neural networks \\ for
simulating solitary waves of nonlinear wave equations} \\[0.2in]
\end{center}

\begin{center}
Jin Song$^{1,2}$, Ming Zhong$^{1,2}$, George Em Karniadakis$^{3}$ and Zhenya Yan$^{1,2,*}$\footnote{{\it Email address}: zyyan@mmrc.iss.ac.cn (Corresponding author)}  \\[0.15in]
{\it \small $^1$KLMM, Academy of Mathematics and Systems Science, Chinese Academy of Sciences, Beijing 100190, China\\
$^2$School of Mathematical Sciences, University of Chinese Academy of Sciences, Beijing 100049, China\\
$^3$Division of Applied Mathematics, Brown University, Providence, RI, 02912, USA}
\end{center}

\vspace{0.1in} \noindent {\bf Abstract:}\, We propose a new two-stage initial-value iterative neural network (IINN) algorithm for solitary wave computations of nonlinear wave equations
based on traditional numerical iterative methods and physics-informed neural networks (PINNs). Specifically, the IINN framework consists of two subnetworks, one of which is used to fit a given initial value, and the other incorporates physical information and continues training on the basis of the first subnetwork. Importantly, the IINN method does not require any additional data information including boundary conditions, apart from the given initial value. Corresponding theoretical guarantees are provided to demonstrate the effectiveness of our IINN method.  The proposed IINN method is efficiently applied to  learn some types of solutions in different nonlinear wave equations, including the one-dimensional (1D) nonlinear Schr\"{o}dinger equations (NLS) equation (with and without potentials), the 1D saturable NLS equation with $\PT$-symmetric optical lattices,
the 1D focusing-defocusing coupled NLS equations, the KdV equation, the two-dimensional (2D) NLS equation with potentials, the 2D amended GP equation with a potential, the (2+1)-dimensional KP equation,
and the 3D NLS equation with a potential.
These applications serve as evidence for the efficacy of our method.
Finally, by comparing with the traditional methods, we demonstrate the advantages of the proposed IINN method.

\vspace{0.1in} \noindent  {\it Keywords:} \,
Nonlinear wave equations,\,
Physics-informed deep learning,\,
Initial-value iterative neural network,\,
Solitary waves



\vspace{0.1in}


\baselineskip=13pt

\section{Introduction}\label{sec1}

Solitary waves, discovered and named by Russell in 1833-1834, play an important role in the study of nonlinear wave equations, which can describe shallow water wave mechanics and light propagation in nonlinear photonic lattices \cite{solitary,op1,op2,op3}. For example, the Korteweg–de Vries (KdV) equation, a mathematical physical model for waves on shallow water surfaces, demonstrates several characteristics expected of an integrable partial differential equation (PDE). It possesses a wide range of explicit solutions, particularly soliton solutions, which can be obtained analytically by using the inverse scattering transform (IST)~\cite{GGKM}. Moreover, many other nonlinear integrable PDEs can also be analytically solved to find their solitons via the IST~\cite{ab,Ablowitz-1,Ablowitz-2}.
However, in the case of most non/nearly-integrable nonlinear PDEs, their analytical solution expressions are not available, and numerical computations are required to study these nonlinear waves (e.g.,  solitary waves)~\cite{op2,yang}.

Various traditional numerical methods have been developed to tackle this challenge. One classical approach is the shooting method, which involves reducing a boundary value problem to an initial value problem \cite{shoot1,shoot2}. It involves finding solutions to the initial value problem for different initial conditions until one finds the solution that also satisfies the boundary conditions of the boundary value problem. While the shooting method is effective for solving 1D problems, it is not applicable in higher dimensions.
Another important class of methods are iterative methods, such as the Petviashvili method \cite{yang}, imaginary-time evolution method (ITEM) \cite{it1,it2,it}, squared-operator iteration method (SOM) \cite{som} and Newton’s method \cite{netwon,NCG}.
In this category of methods, the solution is updated by a fixed iterative scheme. For instance, in ITEM, a solitary wave with a specified power is sought by numerically integrating the underlying nonlinear wave equation with the evolution variable $t$ replaced by $it$ \cite{it}, where $i$ is the imaginary unit.
In Newton’s method, the solution is updated by solving a linear inhomogeneous operator equation, where the inhomogeneous term is the residue of the nonlinear wave equation \cite{netwon2}.
Moreover, the conjugate-gradient method is applied in solving this linear equation not by direct methods as in the traditional Newton’s method, which speeds up the convergence considerably \cite{NCG}. These numerical methods can achieve fast convergence.
However, the choice of discretization scheme has a significant impact on the algorithm’s accuracy and implementation difficulty. In general, the finite-difference discretization is commonly used, while it has a low accuracy compared to spectral method.
In particular, for high-dimensional problems, the memory required by finite difference methods will increase exponentially.
Spectral methods, known for their high accuracy and fast convergence speed, also have drawbacks, such as the inability to adapt to complex computational domains \cite{spectral}; spectral elements methods can be employed but they are relatively complex to implement \cite{Karniadakis_Sherwin_book}.
Therefore, there is an urgent need for a new and efficient method that can handle high-dimensional and complex regional problems, offering fast convergence and easy implementation.

Recently, the remarkable progress in machine learning has revolutionized various scientific fields, such as image recognition, natural language processing, cognitive science, data assimilation, and many others \cite{DL,aml,Imagenet classification,app2,app3,app4,app5}.
These advancements have been made possible by the rapid expansion of computing resources.
Especially, one emerging subfield in machine learning is the use of deep learning to solve PDEs under the concept of Scientific Machine Learning (SciML) \cite{sml}.
Neural networks (NNs) are capable of approximating solutions to PDEs based on the universal approximation theorems \cite{universal app}.
This has led to the development of various methods for solving PDEs using NNs \cite{ewn,relunn,jam,palr,pinn,deepxde}.
For example, in Ref.~\cite{ewn}, the variational form of PDEs is adopted, and by minimizing the corresponding energy functional the solution can be obtained.
In particular, with the aid of automatic differentiation \cite{ad1,ad2}, one could consider a residual term from the given PDEs in strong form directly.
These methods, known as physics-informed neural networks (PINNs) \cite{pinn,deepxde}, leverage automatic differentiation to avoid truncation errors and numerical quadrature errors of variational forms.
Compared to traditional mesh-based methods like the finite difference method and spectral methods mentioned earlier, deep learning offers a mesh-free approach by taking advantage of automatic differentiation, and could overcome the curse of dimensionality \cite{d1}.
Based on these advantages, the NN methods have also been applied extensively to different types of PDEs, and variants and extensions targeted at different application scenarios have also
subsequently emerged  \cite{NRP21,JK,WYP,spde,fpde,bpinn,two,pu2021solving,li2022mix,wu2022prediction,wang2021datadriven,zhong2022data1,song2021deep}.
However, most of these NN methods focused on solving the initial-boundary value problems, where the solutions of the equations were uniquely determined and the initial-boundary value conditions needed to be taken into account in the loss functions.
If these methods are directly applied to solitary wave computations, it is difficult to obtain the desired solution since there are many solitary wave solutions satisfying the physical constraints.
To the best of our knowledge, there is no effective research on the required solitary waves for the multi-solution problems of nonlinear wave equations by deep learning methods.
Whether deep learning methods can be effectively used to solve multi-solution problems is still an unknown and significant topic.


To fill in the gap, a novel algorithm called the two-stage initial-value iterative neural network (IINN) is proposed herein for solitary wave computations of nonlinear wave equations, whose ideas combine numerical iterative methods with PINNs.
IINN consists of two subnetworks. We first choose an appropriate initial value such that the first subnetwork approximates it sufficiently, which resembles the concept of given initial values in iterative methods.
Then, we initialize the parameters in the second subnetwork with the learned weights and biases from the first network, and consider the PDE residual in the loss function and minimize it.
In other words, we continue to optimize our network based on the given initial values such that the output satisfies the given equation.
From the machine learning perspective, the approach is known as transfer learning, where knowledge gained from training one model is transferred to another model, typically when the two models have similar tasks or domains.
One key advantage of the IINN method is that it does not require additional data information including boundary conditions, apart from the given initial value. This significantly reduces the difficulty of network optimization.
The effectiveness of the proposed method is supported by corresponding theoretical guarantees. In addition, the IINN has demonstrated robustness and convergence in various numerical testings involving different physical wave systems and a wide range of initial conditions, as long as the initial condition is reasonably close to the exact solution.

The remainder of this paper is arranged as follows.  Preliminaries including some notations, definitions and known methods are given in Sec.~2.
Then in Sec.~3, we propose the IINN algorithm and provide corresponding theoretical guarantees.
We demonstrate the performance of the IINN method by applying it to various examples of solitary wave computations of many types of nonlinear wave equations, especially for high-order/higher-dimensional nonlinear wave equations. We also present the comparison between the IINN method and traditional methods in Sec.~4. Finally, some conclusions and discussions are given in Sec.~5.

\section{Preliminaries}

\subsection{Some notations and definitions}


\bf{Notations:} \rm Let $\mathbb{R}^d$ and $\mathbb{C}^m$ be the $d$-dimensional real space and $m$-dimensional complex space, respectively. $|\cdot|$ stands for the $L_1$ norm for scalar or vector, and $\|\cdot\|_2$ represents the $L_2$ norm. $\mathcal{C}^2$ represents the space of twice-continuously differentiable functions. $\nabla$ is gradient operator and $\Delta$ is Laplace operator. $\mathrm{Pr}(A)$ represents the probability of event $A$ occurring. Let $p(a,b)=\|a-b\|_2$ be the distance between vectors $a$ and $b$. Let $d(a,V)=\min_{v\in V}p(a,v)$ represent the distance between vector $a$ and set $V$. Suppose vector function $f$ is twice differentiable, then let $H(f)$ represent the Hessian matrix of $f$. Let $\lambda_{\min}(\cdot)$ be the minimum eigenvalue of a matrix.

\begin{definition}
  Let $f$ denote a real-valued function in domain $\Omega$. Then, $f$ has Lipschitz continuity if there exists constant $\rho>0$ and $f$ satisfies
  \begin{equation}\label{lp}
    \|\nabla f(x)-\nabla f(y)\|_2\leq \rho\|x-y\|_2, \quad \mathrm{for}\,\,\mathrm{any}\,\, x, \,y\in\Omega.
  \end{equation}
\end{definition}

\begin{definition}
  Let $x\in\mathbb{R}^d$ be a vector and $\Lambda\subseteq\mathbb{R}^d$ be a set. Then, $x$ is isolated in $\Lambda$ if there is a neighborhood $U\subseteq\mathbb{R}^d$ around $x$, and $U\cap \Lambda=\oslash$.
\end{definition}

\begin{definition}
  Suppose that $f$ is twice differentiable, then $x^*$ is a strict saddle if $\lambda_{\min}(H(f)_{x=x^*})<0$.
\end{definition}

\subsection{Problem statement}
The problem we are interested in is the computation of solitary waves in a general nonlinear wave system in arbitrary spatial dimensions, which are special localized solutions that maintain their shapes as they propagate. The system can be written in the following form:
\begin{equation}\label{L0}
  \mathbf{L_0}\mathbf{u}(\mathbf{x})=0,
\end{equation}
where $\mathbf{L_0}$ is a nonlinear operator, $\mathbf{x}=(x_1,x_2,\cdots, x_d)\in\mathbb{R}^d$ is a vector spatial variable, $\mathbf{u}(\mathbf{x})\in\mathbb{C}^m$ is a complex-valued vector solitary wave solution, and $\mathbf{u}\rightarrow 0$ as $|\mathbf{x}|\rightarrow\infty$. During practical computations, it is common to restrict $\mathbf{x}$ to a sufficiently large finite domain, that is, $\mathbf{x}\in\Omega\subseteq \mathbb{R}^d$, and $\mathbf{u}(\mathbf{x})\to 0$ when $\mathbf{x}\to \partial \Omega$.
For example, the $d$-dimensional scalar generalized nonlinear Schr\"{o}dinger (NLS) equation with a potential has the following form:
\begin{equation}\label{U}
  iU_t-\Delta U + V(\mathbf{x}) U + \mathcal{N}(\mathbf{x}, |U|^2)U=0,
\end{equation}
where {\color{red}$U=U(\mathbf{x},t)\in\mathbb{C}$} is a complex field of the $d$-dimensional spatial variable $\mathbf{x}\in\mathbb{R}^d$ and time $t$, $\Delta=\partial_{x_1}^2+\partial_{x_2}^2+\cdots+\partial_{x_d}^2$ a $d$-dimensional Laplacian, $V(\mathbf{x})$ a real or complex potential, and $\mathcal{N}(\mathbf{x}, |U|^2)$ a function of $\mathbf{x}$ and intensity $|U|^2$.
The stationary solitary waves (e.g., ground states and excited states) of this equation can be written in the form
\begin{equation}\label{solu}
  U(\mathbf{x},t)=u(\mathbf{x})e^{i\mu t},
\end{equation}
where $u(\mathbf{x})$ is a complex and localized function and $\mu\in\mathbb{R}$ is the propagation constant.
Substituting it into Eq.~(\ref{U}) yields the stationary differential equation for $u(\mathbf{x})$
\begin{equation}\label{lu}
    L u =0, \quad  \mathrm{where} \quad  L= -\Delta + V(\mathbf{x}) + \mathcal{N}(\mathbf{x}, |u|^2)-\mu.
\end{equation}
Notice that when $m=1$ ($\mathbf{u}(\mathbf{x})\in\mathbb{C}$), we denote $\mathbf{L_0}\mathbf{u}$ as $L u$, and the same applies in later examples.
Eq.~(\ref{lu}) admits solitary waves of various forms for a large class of functions $\mathcal{N}(\mathbf{x}, |u|^2)$ and potential $V(\mathbf{x})$. Especially, for the same equation, there may be different forms of solitary waves (such as trivial solution, degenerate state, symmetry breaking bifurcations, and so on). In the area of scientific computing, especially in solving forward problems of nonlinear partial differential equations, the discovery of solitary waves is still an open problem (see, e.g., Ref.~\cite{yang} and reference therein).

\subsection{Traditional numerical methods}

Numerous numerical methods have been developed thus far to compute solitary waves of nonlinear wave euqations. One is the shooting method, which is efficient for 1D problems but does not apply in higher dimensions \cite{shoot1,shoot2}.
Other methods commonly used for solving these problems are iterative methods in the scheme $\mathbf{u}_{n+1}=\mathcal{M}_n\mathbf{u}_n$ for given initial state $\mathbf{u}_0$ with iterative operator $\mathcal{M}_n$,
including the Petviashvili method, accelerated imaginary-time evolution (AITEM) method, squared-operator iteration (SOM) method and Newton-conjugate-gradient (NCG) method \cite{yang,pi,it,som,NCG}. Among them, both the Petviashvili method and the AITEM can only converge to the ground states of nonlinear wave equations and would diverge for excited states. For multi-component equations, they may even diverge for the ground states.
Furthermore, these iterative methods require that the initial conditions are sufficiently close to the desired exact solutions in order to guarantee algorithm convergence.

Although SOM and NCG methods have higher convergence speed and convergence rates, they are more challenging to operate when dealing with high-dimensional problems and specific region problems.
Taking into account the finite difference methods or spectral methods used in discretizing derivatives, meshing is necessary to discretize the domain, which can lead to exponential growth in the amount of storage required and the complexity of computations.

\subsection{The PINNs method}

Recently, being different from traditional numerical method, deep neural networks were introduced to approximate the solution of partial differential equations (PDEs) with the aid of automatic differentiation methods, which reduce the cost of constructing computationally-expensive grids. Especially, the physics-informed neural networks (PINNs) approach \cite{pinn} was used to consider the important physical laws given by the PDEs to control the output solution of a deep neural network, which significantly reduces the required amount of data.
The PINNs framework for the data-driven solutions of nonlinear systems (\ref{L0}) can be introduced as follows.

Firstly, a fully-connected neural network $\mathrm{NN}(\mathbf{x}; \theta)$ with $n$ hidden layers and $m$ neurons in each layer is constructed to learn the solution $\mathbf{u}(\mathbf{x})$, where the parameters $\theta=\{W,B\}$ with $W = \{w_j\}_{1}^{n+1}$ and $B = \{b_j\}_1^{n+1}$ being the weight matrices and bias vectors, respectively. Then the vector data of the hidden layers
and output layer can be generated by following affine transformation $\mathcal{F}_j$
\begin{equation}\label{sigma}
\begin{array}{l}
 A_j=\sigma\left(\mathcal{F}_j(A_{j-1})\right)=\sigma(w_j\cdot A_{j-1}+b_j),\quad j=1,2,...,n,\quad
  A_{n+1}=\mathcal{F}_{n+1}(A_{n})=w_{n+1}\cdot A_{n}+b_{n+1},
  \end{array}
\end{equation}
where $\sigma(\cdot)$ denotes some nonlinear activation function, $w_j$ is a dim$(A_j)\times $dim$(A_{j-1})$ matrix, $A_0=\mathbf{x}$, and $A_j=(a_{j1},...,a_{jm})^T$, $b_j=(b_{j1},...,b_{jm})^T$. Therefore the relation between input $\mathbf{x}$ and output $\hat{\mathbf{u}}(\mathbf{x};\theta)$ is given by
\begin{equation}\label{uu}
  \hat{\mathbf{u}}(\mathbf{x};\theta)= A_{n+1}=\left(\mathcal{F}_{n+1}\circ \sigma\circ\mathcal{F}_{n}\circ\cdots \circ\sigma\circ\mathcal{F}_{1}\right)(\mathbf{x}),
\end{equation}
where the activation function $\sigma$ is chosen as the hyperbolic tangent function $\tanh(\cdot)$ to ensure the smoothness of $\hat{\mathbf{u}}$.

To ensure that the output $\hat{\mathbf{u}}(\mathbf{x};\theta)$ satisfies the equation (\ref{L0}), we utilize the total mean squared error (MSE) to define the following loss function and optimize parameters $\theta$ to minimize the value of loss.
\begin{equation}\label{Loss0}
 \mathcal{L}_0:=MSE_L+MSE_b=\d\frac{1}{N_f}\sum_{\ell=1}^{N_f}|\mathbf{L}_0\hat{\mathbf{u}}(\mathbf{x}_f^{\ell})|^2+\frac{1}{N_b}\sum_{\ell=1}^{N_b}|\hat{\mathbf{u}}(\mathbf{x}_b^{\ell})|^2,
\end{equation}
where $\{\mathbf{x}_f^\ell\}_\ell^{N_f}$ are connected with the randomly chosen sample points in $\Omega$, and $\{\mathbf{x}_b^\ell\}_\ell^{N_b}$ are linked
with the randomly selected boundary points in $\partial\Omega$. With the aid of some optimization approaches (e.g., SGD, Adam \& L-BFGS~\cite{adam,bfgs}), we minimize the loss $\mathcal{L}_0$ to make the learned solution $\hat{\mathbf{u}}(\mathbf{x};\theta)$ satisfy Eq.~(\ref{L0}).

It should be noted that before training an $\mathrm{NN}$ model, the parameters $\theta$ need to be initialized.
In most cases, the bias term is commonly initialized to zero. There are several effective methods available for initializing weight matrices, such as Glorot initialization and He initialization \cite{glorot,he}, which help to address the issue of improper initialization and can improve the performance and convergence of neural networks.

\section{Methodology and applications}

\subsection{Methodology: the IINN framework}

When we try to apply the PINNs~\cite{pinn} to compute the solitary wave solutions of Eq.~(\ref{L0}), the learned results are not satisfactory.
Especially, if we directly apply Eq.~(\ref{Loss0}) as the loss function, the PINNs often converges to a trivial solution.
Even if the network converges to a non-trivial solution, that solution may not be what we desire because the same equation (\ref{Loss0}) can admit different states.
Inspired by traditional numerical iteration methods, we propose the following initial value iterative neural network (IINN) algorithm to solve this problem and provide corresponding theoretical guarantees.

In the following, we will introduce the main idea of IINN method.
Two identical fully connected neural networks $\mathrm{NN}_1$ and $\mathrm{NN}_2$, defined by Eq.~(\ref{uu}), are employed to learn the desired solution $\mathbf{u}^*$.

{\it \bf NN$_1$}---First, we choose an appropriate initial value $\mathbf{u}_0$ such that it is sufficiently close to $\mathbf{u}^*$. Then we randomly select $N$ training points
$\{\mathbf{x}_i\}_{i=1}^{N}$ within the region $\Omega$ and train the network parameters $\theta$ by minimizing the mean squared error loss $\mathcal{L}_1$, aiming to make the output of $\mathrm{NN}_1$ $\bar{\mathbf{u}}$ sufficiently close to initial value $\mathbf{u}_0$, where loss function $\mathcal{L}_1$ is defined as follows
\begin{equation}\label{Loss1}
  \mathcal{L}_1:=\frac{1}{N}\|\bar{\mathbf{u}}-\mathbf{u}_0\|_2^2=\frac{1}{N}\sum_{i=1}^{N}|\bar{\mathbf{u}}(\mathbf{x}_i)-\mathbf{u}_0(\mathbf{x}_i)|^2.
\end{equation}

{\it \bf NN$_2$}---Then, we initialize the parameters $\theta$ of $\mathrm{NN}_2$ with the learned weights and biases from $\mathrm{NN}_1$,
that is
\begin{equation}\label{argmin}
  \theta_0=\mathrm{argmin}\, \mathcal{L}_1(\theta).
\end{equation}
For the output of $\mathrm{NN}_2$ $\hat{\mathbf{u}}$,
we define the loss function $\mathcal{L}_2$ as follows and utilize SGD or Adam optimizer~\cite{adam} to minimize it.
\begin{equation}\label{Loss2}
  \mathcal{L}_2:=\frac{1}{N}\frac{\|\mathbf{L_0}\hat{\mathbf{u}}\|_2^2}{\max(|\hat{\mathbf{u}}|)}=\frac{1}{N}\frac{\sum_{i=1}^{N}|\mathbf{L_0}\hat{\mathbf{u}}(\mathbf{x}_i)|^2}{\max_i(|\hat{\mathbf{u}}(\mathbf{x}_i)|)}.
\end{equation}
It should be noted that $\mathcal{L}_2$ is different from the loss function $\mathcal{L}_0$ defined in PINNs. Here we are not taking boundaries into consideration, instead we incorporate $\max(|\hat{\mathbf{u}}|)$ to ensure that $\hat{\mathbf{u}}$ does not converge to trivial solution.

Based on the above introduction, the framework of IINN algorithm is summarized in Alg.~\ref{alg1}. Meanwhile, to provide clarity, Fig.~\ref{net} shows the schematic representation of the IINN method.

\begin{algorithm}
    \caption{The framework of initial value iterative neural network (IINN)}\label{alg1}
    \begin{algorithmic}
        \REQUIRE Operator $\mathbf{L_0}$ in (\ref{L0}); initial state $\mathbf{u}_0$; error threshold $\varepsilon_1$ and $\varepsilon_2$; training data $\{\mathbf{x}_i,\mathbf{u}_0(\mathbf{x}_i)\}_{i=1}^{N}$; learning rate $\alpha$, maximum iteration number $K$.
        \ENSURE Output $\hat{\mathbf{u}}$.
        \STATE For $\mathrm{NN}_1$, randomly initialize the parameters $\theta_0$ s.t. they satisfy the normal distribution. For network output $\bar{\mathbf{u}}(\mathbf{x},\theta_0)$, $\mathcal{L}_1:=\frac{1}{N}\|\bar{\mathbf{u}}-\mathbf{u}_0\|_2^2$.
    \FOR{$k=0:K$}
        \IF{$\mathcal{L}_1(\theta_k) < \varepsilon_1$}
            \STATE $\theta^*=\theta_k$;
            \STATE break;
        \ELSE
            \STATE Apply the Adam optimizer update parameters $\theta_{k}$;
        \ENDIF
    \ENDFOR
    \STATE For $\mathrm{NN}_2$, initialize the parameters $\theta_0=\theta^*$ and set $k=0$. For network output $\hat{\mathbf{u}}(\mathbf{x},\theta_0)$, $\mathcal{L}_2:=\frac{1}{N}\frac{\|\mathbf{L_0}\hat{\mathbf{u}}\|_2^2}{\max(|\hat{\mathbf{u}}|)}$.
    \WHILE{$\mathcal{L}_2 \geq \varepsilon_2$}
        \STATE $\theta_{k+1}=\theta_k-\alpha \nabla\mathcal{L}_2$;
        \STATE $k=k+1$;
    \ENDWHILE
    \end{algorithmic}
\end{algorithm}

\begin{figure*}[!t]
    \centering
  {\scalebox{0.67}[0.67]{\includegraphics{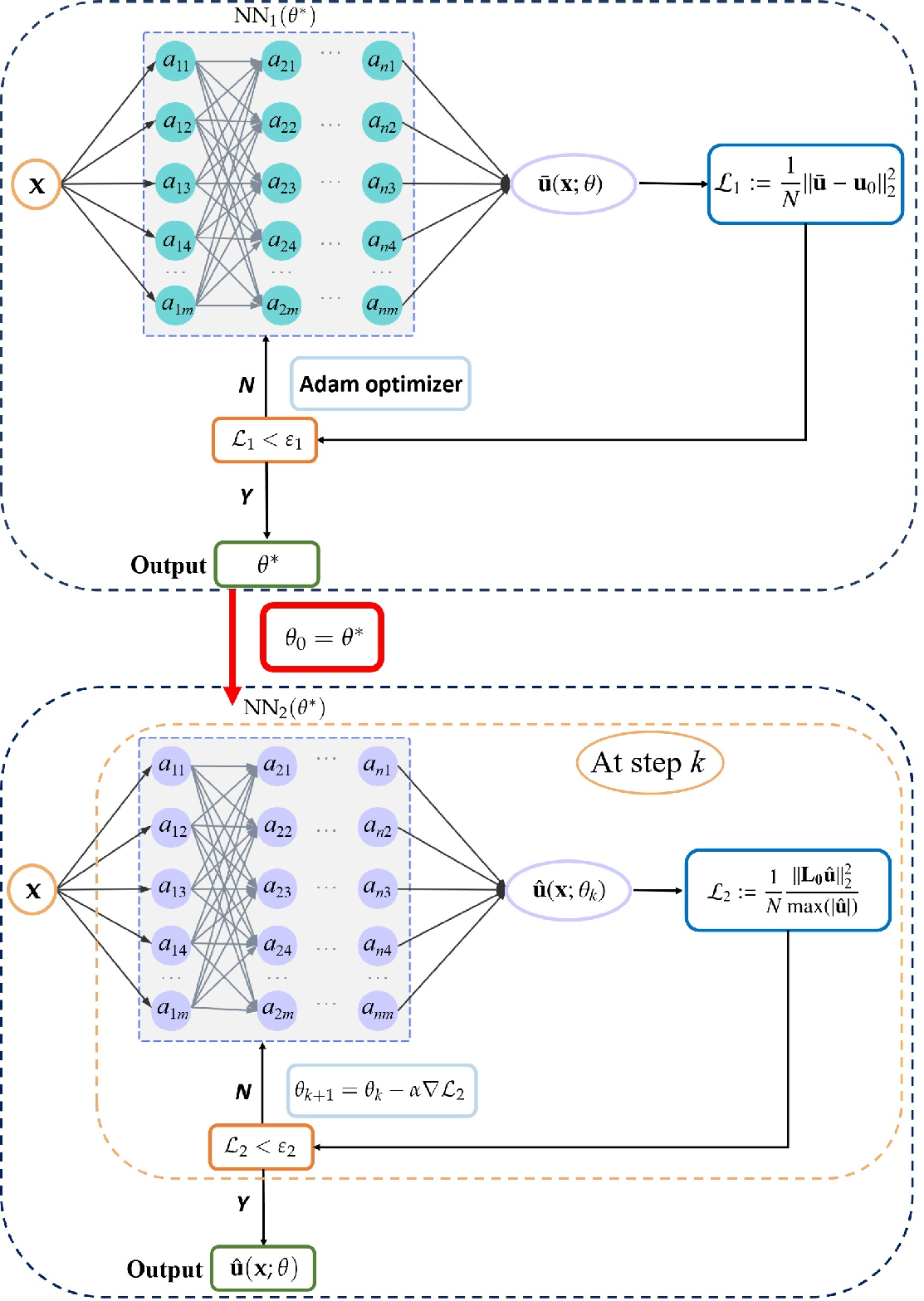}}}
  \vspace{0.1in}
\caption{\small \rm The schematic diagram of IINN method.}
  \label{net}
\end{figure*}

\begin{remark}
The core idea of the IINN algorithm is to initialize the parameters of $\mathrm{NN}_2$ with the learned weights and biases from $\mathrm{NN}_1$,
that is, $\theta_0=\mathrm{argmin}\, \mathcal{L}_1(\theta)$.
From the perspective of numerical iteration, for $\mathrm{NN}_2$, we iterate the network parameters with $\hat{\mathbf{u}}(\theta_0)$ as the initial value, such that $\hat{\mathbf{u}}$ satisfies the Eq.~(\ref{L0}) by minimizing loss function $\mathcal{L}_2$.
From a machine learning perspective, the approach is known as transfer learning, where knowledge gained from training one model is transferred to another model, typically when the two models have similar tasks or domains. By initializing $\mathrm{NN}_2$ with the parameters of $\mathrm{NN}_1$, we can leverage the pre-trained model’s learned representations and potentially achieve better performance, especially if the new task or data is related to the original task or data on which $\mathrm{NN}_1$ is trained.
\end{remark}

\begin{remark}
In the scheme of IINN, the choice of initial value $\mathbf{u}_0$ is crucial as it determines the type of solution we ultimately obtain.
There are some techniques for selecting initial states.
Usually, based on the characteristics of the system and our understanding of the system, we can estimate the initial value using physical background knowledge or past experience.
For example, 1D NLS admits the $\mathrm{sech}$-type soliton solution. Therefore, we can take $u_0(x)=A\mathrm{sech}(x)$, then by adjusting the coefficient $A$ to make $|Lu_0|$ smaller than a certain threshold. Furthermore, external potentials play a crucial role in solitons shaping and solitons management. For example, for the 1D NLS with harmonic-oscillator (HO) trapping potential, the exact solution is close to $\exp(-x^2)$. Thus, we choose $u_0=A\exp(-x^2)$ as initial state and adjust $A$ to make $u_0$ close to exact solution sufficiently.
More generally, we can obtain the initial conditions by computing the spectra and eigenmodes in the linear regime. Taking Eq.~(\ref{lu}) as an example, when the power of $u$ is small (that is $\|u\|_2^2$ is small), we can consider that the solutions of Eq.~(\ref{lu}) is the eigenmodes of the linear problem (Eq.~(\ref{lu}) in the absence of the nonlinearity $\mathcal{N}(\mathbf{x}, |u|^2)$). Then solutions to Eq.~(\ref{lu}) for values of the propagation constant $\mu$ taken in a vicinity of linear eigenvalues may be approximated by eigenfunctions.
\end{remark}

In the following, the corresponding theoretical evidence is provided to guarantee the effectiveness of IINN method.
First, a lemma is presented as follows.
\begin{lemma}\label{lemma1}
  Let $\Lambda=\bigcup_{i=1}^{N} \Lambda_i$, where $\Lambda_i=\left\{\theta_i\,|\, \mathbf{L_0}\hat{\mathbf{u}}(\theta_i)=0\right\}$ and for any $\theta_i^m$, $\theta_i^n\in\Lambda_i$,
  $\|\hat{\mathbf{u}}(\theta_i^m)-\hat{\mathbf{u}}(\theta_i^n)\|_2=0$,
  and $N$ is the number of distinct solitary wave solutions. Then, $\theta_i\in\Lambda_i$ is isolated in $\Lambda_j$ for $i\neq j$.
\end{lemma}
\begin{proof}
  First according to definition (\ref{uu}) of $\mathrm{NN}_2$, $\hat{\mathbf{u}}(\theta)$ is continuous with respect to $\theta$. For $\theta_i\in\Lambda_i$, we denote
  \bee
  d=\min_{\theta_j\in\Lambda_j,j\neq i}\|\hat{\mathbf{u}}(\theta_j)-\hat{\mathbf{u}}(\theta_i)\|_2.
  \ene
  It is obvious that $d>0$.
  Considering the continuity of $\hat{\mathbf{u}}$, there exists a neighborhood $U$ around $\theta_i$, such that any $\theta\in U$, $p(\hat{\mathbf{u}}(\theta),\hat{\mathbf{u}}(\theta_i))<d/2$.
  Therefore, $U\cap \Lambda_j=\oslash$ for $j\neq i$. The proof is completed.
\end{proof}

According to the stable manifold theorem from dynamical systems theory \cite{shub}, we provide the following important theorem, and its proof can be found in Ref.~\cite{the1}.
\begin{theorem}\cite{the1}
 If $f$ is a $\mathcal{C}^2$ function   and has Lipschitz continuity as defined in Definition 1,  and $\theta^*$ be a strict saddle. Assume that learning rate $0<\alpha< \frac{1}{\rho}$, then
  \begin{equation}
    \mathrm{Pr}\left(\lim_k\theta_k=\theta^*\right)=0.
  \end{equation}
\end{theorem}\label{theorem1}

Finally, with the aid of above analysis, we provide the following theorem that guarantees the effectiveness of IINN method.

\begin{theorem}
  For a given soliton state $\mathbf{u}^*$, suppose that the initial state $\mathbf{u}_0$ is sufficiently close to $\mathbf{u}^*$, and the output of $\mathrm{NN}_1$ $\bar{\mathbf{u}}(\theta^*)=\mathbf{u}_0$.  And $\theta^*$ satisfies $d(\theta^*,\Lambda_i)<d(\theta^*,\Lambda_j)$ $(j\neq i)$,
  where $i$ satisfies $\hat{\mathbf{u}}(\theta_i)=\mathbf{u}^*$. Then, the output of $\mathrm{NN}_2$ $\hat{\mathbf{u}}$ is sufficiently close to $\mathbf{u}^*$, for sufficiently small learning rate $\alpha$.
\end{theorem}
\begin{proof}
  First, according to Lemma \ref{lemma1}, for any $\theta_i^m\in\Lambda_i$, there exists a neighborhood $U_m$ around $\theta_i^m$, such that for any $j\neq i$, $U_m\cap \Lambda_j=\oslash$.
  Since $\mathbf{u}_0$ is sufficiently close to $\mathbf{u}^*$ and $\theta^*$ satisfies $d(\theta^*,\Lambda_i)<d(\theta^*,\Lambda_j)$ $(j\neq i)$, then $\theta^*\in\cup U_m$.
   Then by SGD algorithm and Theorem 1, for $\mathrm{NN}_2$, $\theta$ almost surely converge to certain $\theta_i\in\Lambda_i$ for sufficiently small learning rate $\alpha$.
   Therefore, the output of $\mathrm{NN}_2$ $\hat{\mathbf{u}}(\theta_i)$ is sufficiently close to $\mathbf{u}^*$.
\end{proof}

\begin{remark}
In practical applications, we often use the Adam optimizer instead of SGD optimizer. The Adam optimizer introduces a momentum term, which accelerates the parameter update process and helps escape local minima, making it more likely to find better convergence points.
\end{remark}

In order to evaluate the performance of the IINN method, we introduce the relative $L_2$ error $E_1$ between the exact solution $\mathbf{u}^*$ and the learned one $\hat{\mathbf{u}}$ on $\mathbf{x}$ grids, where
\begin{equation}\label{E1}
  E_1=\frac{\|\hat{\mathbf{u}}(\mathbf{x})-\mathbf{u}^*(\mathbf{x})\|_2}{\|\mathbf{u}^*(\mathbf{x})\|_2}.
\end{equation}

\subsection{Examples}

In this section, we will demonstrate the performances of the IINN method by applying it to various examples of solitary wave computations.
For the following example, if not otherwise specified, we choose a 4-hidden-layer deep neural network with 100 neurons per layer, and set learning rate $\alpha=0.0001$.
In the case of certain specific systems, we can find exact soliton solutions, which will serve as a benchmark to evaluate the performance of the network by calculating $E_1$.
For general cases, we utilize numerical methods such as the
Newton-conjugate-gradient (NCG) method \cite{NCG} to obtain high-precision approximate solutions, which serve as a reference for comparison.

Generally speaking, due to the automatic differentiation algorithm, training $\mathrm{NN}_2$ usually takes much more time than training $\mathrm{NN}_1$, especially for high-order equations and high-dimensional systems.
Therefore, to ensure convergence speed, the number of training iterations for $\mathrm{NN}_1$ needs to be large enough, or the threshold $\varepsilon_1$ needs to be small enough.
In the following, we set error threshold $\varepsilon_1$=1e-07.
Finally, we should mention that all computations are performed by using a Lenovo notebook with a 2.30GHz eight-cores i7 processor and a RTX3080 graphics processor.


\v \noindent {\bf Example 3.1} (Solitons of the 1D NLS equation with Kerr nonlinearity). The first example we consider is the 1D NLS equation with Kerr nonlinearity (where $\mathcal{N}(x,|U|^2)U$ in Eq.~(\ref{U}) is taken as the Kerr nonlinear term $g|U|^2U$):
\bee
 iU_t-U_{xx}+V(x)U+g|U|^2U=0,
\ene
where $V(x)$ denotes the potential. The corresponding stationary Eq.~(\ref{lu}) has the following form
\begin{equation}\label{NLS1d}
   L u =0, \quad  L= -\partial_{xx} + V(x) + g|u|^2-\mu.
\end{equation}

In the following, we consider three scenarios: $V = 0$, $V$ taking the form of harmonic-Gaussian (HG) potential and $V$ taking the complex Scarf-II potential.

\v{\it Case 1.}---{\it Bright soliton of the 1D NLS equation with $V=0$ and $g=-1$.} In this case, Eq.~(\ref{NLS1d}) admits the bright soliton as follows
\begin{equation}\label{nls1dsolu}
  u(x)=\sqrt{-2\mu}\,\mathrm{sech}(\sqrt{-\mu}x),\qquad \mu<0.
\end{equation}
Based on IINN method, the initial state is taken
\begin{equation}\label{nlsu0}
  u_0(x)=\mathrm{sech}(x),
\end{equation}
and set $\Omega=[-20,20]$ with $N=500$. Through 10000 steps of iterations, with $\mathrm{NN}_1$ taking 21s and 25000 steps of iterations, with $\mathrm{NN}_2$ taking 199s, the relative $L_2$ error $E_1$=1.637538e-03 compared to the exact solution (\ref{nls1dsolu}) at $\mu=-2$. Figs.~\ref{f-nls}(a1, a2) illustrate the comparison between the learned solutions and exact
solutions at $\mu=-2$ as well as the loss-iteration diagram for $\mathrm{NN}_2$.

\begin{figure*}[!t]
    \centering
  {\scalebox{0.8}[0.8]{\includegraphics{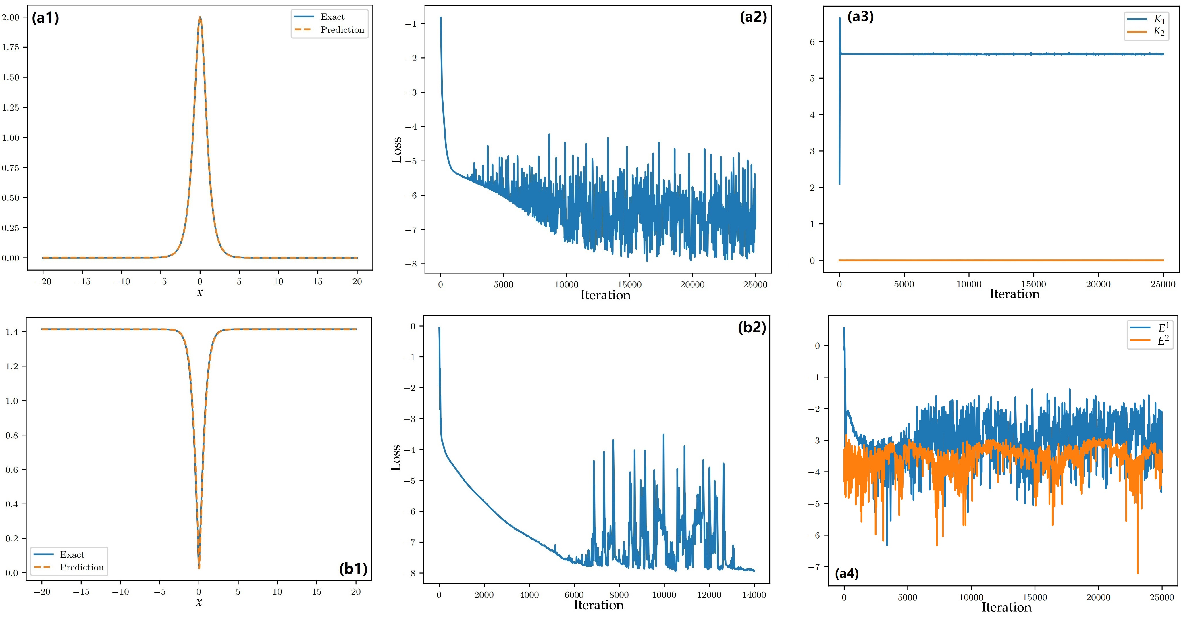}}}
\hspace{-0.4in}\vspace{0.15in}
\vspace{0.1in}\caption{\small \rm The soliton solutions $u(x)$ of 1D NLS equation (\ref{NLS1d}) in free space ($V=0$). (a1) The learned bright soliton solution and exact one at $\mu=-2$ in self-focusing case.
(b1) The learned dark soliton solution and exact one at $\mu=2$ in self-defocusing case.
(a2, b2) The loss-iteration plot of $\mathrm{NN}_2$, where the vertical axis represents $\log_{10}\mathcal{L}_2$ (the same hereinafter).
(a3) The conserved quantities $\int_{\mathbf{R}}\omega\dd x$ versus iteration for bright soliton, where $K_1$ and $K_2$ denote $\omega_1=UU^*$ and $\omega_2=UU_x^*$, respectively.
 (a4) The conserved quantity error versus iteration, where $E^i=\log_{10}|K_i-K_i^*|$ and $K_i^*$ is the true value of conserved quantity.}
  \label{f-nls}
\end{figure*}

Furthermore, considering the integrability of NLS equation, we can verify the accuracy of the solution by checking whether the conserved quantity is constant in every iteration.
A conservation law associated with a differential equation is an expression of the form
\begin{equation}\label{cl}
  \omega_t=J_x,
\end{equation}
where $\omega$ and $J$ are functions of $t$, $x$, $U$ and derivatives of $U$. $\omega$ is called the conserved density and $J$ is called the flux of $\omega$.
The two specific conservation laws for NLS equation are given as follows~\cite{Ablowitz-2}
\begin{equation}\label{nlscl1}
  (UU^*)_t=i(UU_{x}^*-U^*U_x)_x,
\end{equation}
\begin{equation}\label{nlscl2}
  (UU_x^*)_t=i(UU_{xx}^*-U_xU_x^*-\frac{1}{2}gU^2U^{*2})_x,
\end{equation}
Then we can obtain that the conserved quantity is constant when $U\rightarrow 0$ as $|x|\rightarrow\infty$, that is
\begin{equation}\label{cq}
  \frac{\dd }{\dd t}\int_{\mathbf{R}}\omega\dd x=0.
\end{equation}
Since the solution we consider is stationary solution in the form of $U(x,t)=u(x)e^{i\mu t}$, its conserved quantity is invariant over time.
We determine whether the solution converges by examining the change of the conserved quantity during the iteration.
In other words, we detect whether it tends to the true value.
Fig.~\ref{f-nls}(a3) shows these two conserved quantities
\bee
K_1=\int_{\mathbf{R}}\omega\dd x=\int_{\mathbf{R}}UU^*\dd x,\qquad
K_2=\int_{\mathbf{R}}\omega\dd x=\int_{\mathbf{R}}UU_x^*\dd x
\ene
versus iteration for bright soliton, which means the conserved quantities quickly remain the same.
As shown in Fig.~\ref{f-nls}(a4), we also display the variations of the conserved quantity error $E^i$, during the iteration, where $E^i=\log_{10}|K_i-K_i^*|$ and $K_i^*$ is the true value of conserved quantity.
It can be seen that the error tends to be around 1e-3, which is almost of the same order as the above-mentioned relative $L_2$ error $E_1$.

\v{\it Case 2.}---{\it Dark soliton of the 1D NLS equation with $V=0$ and $g=1$.} In this case, Eq.~(\ref{NLS1d}) also has dark soliton solution in self-defocusing case.
\begin{equation}\label{nls1dsolud}
  u(x)=\sqrt{\mu}\,\tanh(\sqrt{\mu/2}x),\qquad \mu>0.
\end{equation}
Although the $|u|\rightarrow A $ as $|x|\rightarrow\infty$, the IINN method is still valid.
Based on the IINN method, we take $\Omega=[-20,20]$ with $N=500$, and take the initial value as
\begin{equation}\label{darku0}
  u_0=\tanh(x).
\end{equation}
Then the learned dark soliton solution can be obtained at $\mu=2$ as shown in Fig.~\ref{f-nls}(b1), after 10000 steps of iterations with $\mathrm{NN}_1$ taking 21s and 14000 steps of iterations with $\mathrm{NN}_2$ taking 110s. The relative $L_2$ error $E_1$=3.408001e-04 compared to the exact solution (\ref{nls1dsolud}) at $\mu=2$. The loss-iteration plot of $\mathrm{NN}_2$ for dark soliton is displayed in Fig.~\ref{f-nls}(b2).

\begin{figure*}[!t]
    \centering
  {\scalebox{0.85}[0.85]{\includegraphics{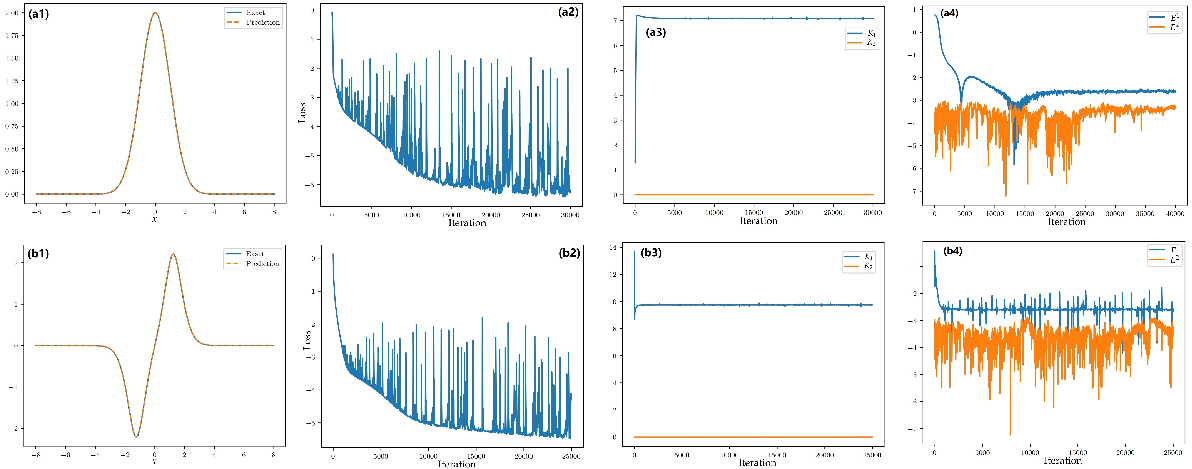}}}
  \hspace{-0.4in}\vspace{0.15in}
\vspace{0.1in}\caption{\small \rm The ground state and dipole mode of 1D NLS equation (\ref{NLS1d}) with HG potential (\ref{HO}). (a1) The learned and exact ground state solution at $V_0=-4$, $\mu=1$ and $g=-1$. (b1) The learned and exact dipole mode at $V_0=-4$, $\mu=1$ and $g=-1$.  (a2, b2) The loss-iteration plot of $\mathrm{NN}_2$ for ground state and dipole mode, respectively.
(a3, b3) The conserved quantities $\int_{\mathbf{R}}\omega\dd x$ versus iteration for ground state and dipole mode respectively, where $K_1$ and $K_2$ denote $\omega_1=UU^*$ and $\omega_2=UU_x^*$, respectively.
 (a4, b4) The conserved quantity error versus iteration for ground state and dipole mode respectively, where $E^i=\log_{10}|K_i-K_i^*|$ and $K_i^*$ is the true value of conserved quantity.}
  \label{f-ho}
\end{figure*}

\v{\it Case 3.}---{\it Ground state and dipole mode (excited state) of the 1D NLS equation with harmonic-Gaussian (HG) potential.}
If we consider the potential $V(x)$ as the HG potential
\begin{equation}\label{HO}
  V(x)=x^2-V_0e^{-x^2},\quad V_0\in\mathbb{R},
\end{equation}
then the exact ground state solution of Eq.~(\ref{NLS1d}) can be found  under self-focusing $(g=-1)$ and self-defocusing $(g=1)$ nonlinearity
\begin{equation}\label{hg1dsolu}
  u(x)=\sqrt{V_0/g}\,e^{-x^2/2},\qquad \mu=1,\qquad V_0/g>0.
\end{equation}

Using the IINN method, we set $\Omega=[-8,8]$ with $N=200$, and take the initial value as
\begin{equation}\label{hgu0}
  u_0(x)=\exp(-x^2).
\end{equation}
Then the learned ground state solution can be obtained at $V_0=-4$ in self-focusing $(g=-1)$ case as shown in Fig.~\ref{f-ho}(a1), after 5000 steps of iterations with $\mathrm{NN}_1$ taking 11s and 30000 steps of iterations with $\mathrm{NN}_2$ taking 238s. The relative $L_2$ error $E_1$=2.511470e-04 compared to the exact solution (\ref{hg1dsolu}) at $V_0=-4$. The loss-iteration plot of $\mathrm{NN}_2$ for ground state is displayed in Fig.~\ref{f-ho}(a2).

Furthermore, in the self-focusing $(g=-1)$ case, Eq.~(\ref{NLS1d}) with HG potential (\ref{HO}) admits the dipole mode, while the exact expression for the solution has not been found yet. Therefore, we utilize the NCG methods to obtain high-precision approximate solutions, which can be referred to as the `exact' solution for comparison purposes.
The derivative discretization scheme is Fourier spectral method \cite{spectral} with the 256 Fourier modes. The computational domain is  discretized by 256 points along each dimension.

Then we choose the initial value as
\begin{equation}\label{hgeu0}
  u_0(x)=4x\exp(-x^2/2).
\end{equation}
Through the IINN method, the learned dipole mode can be obtained at $V_0=-4$ in self-focusing $(g=-1)$ case as shown in Fig.~\ref{f-ho}(b1), after 10000 steps of iterations with $\mathrm{NN}_1$ taking 24s and 25000 steps of iterations with $\mathrm{NN}_2$ taking 201s. The relative $L_2$ error is $E_1$=4.664602e-04. The loss-iteration plot of $\mathrm{NN}_2$ for dipole mode is displayed in Fig.~\ref{f-ho}(b2).
Notice that for the real potential $V(x)$, although the equation is non-integrable, the solution still satisfies the conservation laws (\ref{nlscl1}) and (\ref{nlscl2}); the conserved quantities $\int_{\mathbf{R}}\omega\dd x$ versus iteration for ground state and dipole mode are shown in Figs.~\ref{f-ho}(a3, b3), respectively.
And the variations of the conserved quantity error $E^i$, during the iteration for ground state and dipole mode, are displayed in Figs.~\ref{f-ho}(a4, b4), respectively,  where $E^i=\log_{10}|K_i-K_i^*|$ and $K_i^*$ is the true value of conserved quantity.

On the other hand, we can use the eigenmode in the linear regime as the initial value. By the spectral method \cite{spectral}, we can compute the linear eigenvalue problem of Eq.~(16) in the absence of the nonlinearity.
Since the dipole mode originates from the first excited state, we take the first excited state $\psi$ with eigenvalue $\lambda\thickapprox 4.195$ and set $u_{0e}=A\psi$ as the initial value.
Here we take $A=8$ such that $|Lu_{0e}|$ small enough. The initial value $u_0$ (29) and $u_{0e}$ are exhibited in Fig.~\ref{f-hoeig}(a1).
Through the IINN method, the learned dipole mode is shown in Fig.~\ref{f-hoeig}(a2), after 10000 steps of iterations with $\mathrm{NN}_1$ taking 24s and 25000 steps of iterations with $\mathrm{NN}_2$ taking 199s. The relative $L_2$ error is $E_1$=2.966851e-04. And the loss-iteration plot of $\mathrm{NN}_2$ for dipole mode is displayed in Fig.~\ref{f-hoeig}(a3).

\begin{figure*}[!t]
    \centering
  {\scalebox{0.75}[0.75]{\includegraphics{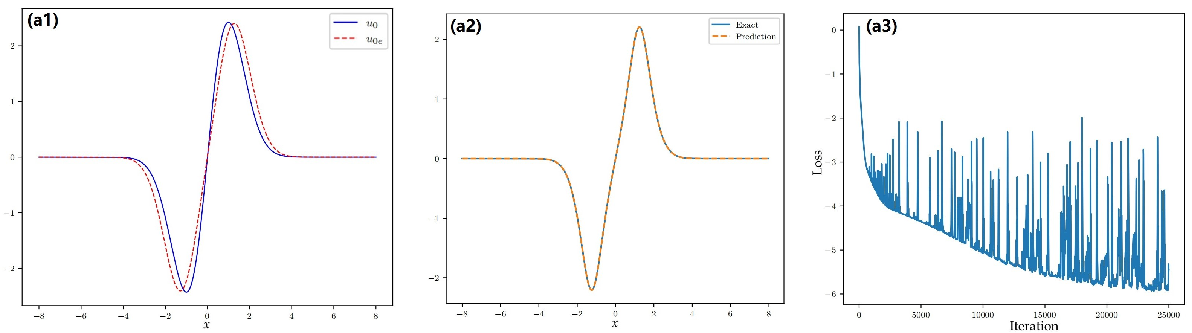}}}
  \hspace{-0.4in}\vspace{0.15in}
\vspace{0.1in}\caption{\small \rm The dipole mode of 1D NLS equation (\ref{NLS1d}) with HG potential (\ref{HO}) using the eigenmodes in the linear regime as the initial value.
(a1) The initial value $u_0$ (\ref{hgeu0}) and $u_{0e}=A\psi$ with the first excited state $\psi$ in the linear regime, where $A=8$. (a2) The learned and exact dipole mode at $V_0=-4$, $\mu=1$ and $g=-1$.
(a3) The loss-iteration plot of $\mathrm{NN}_2$.}
  \label{f-hoeig}
\end{figure*}

\begin{remark}
  It should be noted that the same equation (\ref{NLS1d}) with HG potential (\ref{HO}) has both ground state and dipole solutions.
  Applying the original PINNs method with random initialization method, such as Glorot initialization and He initialization \cite{glorot,he}, it is impossible to obtain two solutions of different forms simultaneously. The network converges to the trivial solution with large probability.
  In the case of multiple solutions, we cannot determine which solution the output of the initialized network is near, and therefore cannot determine which solution the network finally converges to.
  For this reason, the IINN method is proposed for the multi-solution problems of nonlinear wave equations by deep learning methods.
  In fact, we are initializing the network parameters by training the network $\mathrm{NN}_1$ such that the initialized network is relatively close to the target solution.
\end{remark}

\begin{figure*}[!t]
    \centering
  {\scalebox{0.85}[0.85]{\includegraphics{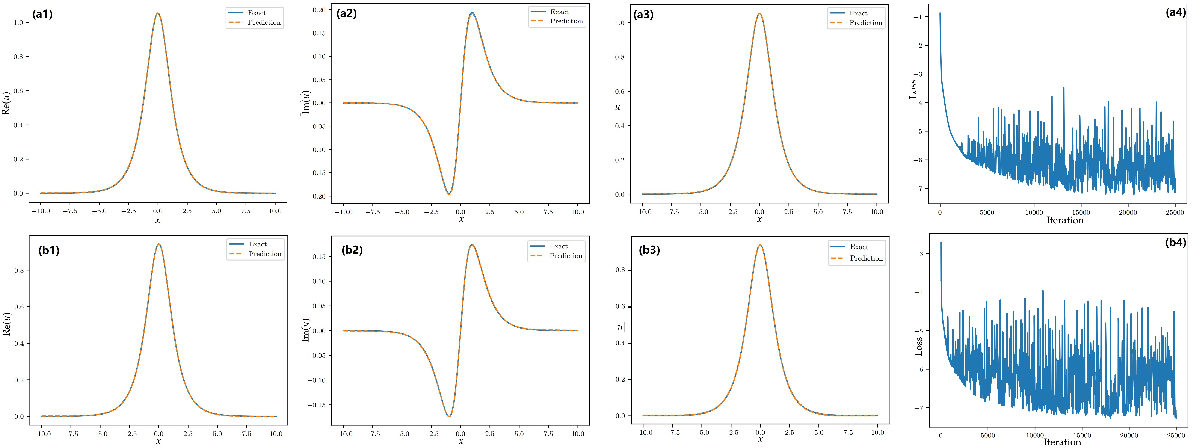}}}
  \hspace{-0.4in}\vspace{0.15in}
\vspace{0.1in}\caption{\small \rm The complex solution $u(x)$ of 1D NLS equation (\ref{NLS1d}) with Scarf-II potential (\ref{scarf}) in self-focusing and self-defocusing cases. In self-focusing case: (a1, a2, a3) The real part, imaginary part and intensity diagrams of learned solution and exact one at $V_0=-1$, $W_0=-1$, $\mu=-1$ and $g=-1$.
 In self-defocusing case: (b1, b2, b3) The real part, imaginary part and intensity diagrams of learned solution and exact one at $V_0=-3$, $W_0=-1$, $\mu=-1$ and $g=1$.
(a4, b4) The loss-iteration plot of $\mathrm{NN}_2$.}
  \label{f-scarf}
\end{figure*}

\v{\it Case 4.}---{\it Soliton solution of 1D NLS equation with complex potentials.} Next we consider solitary waves in the above 1D NLS equation with complex potentials. Especially, the well-known parity-time ($\mathcal{PT}$)-symmetric Scarf-II potential is introduced as follows~\cite{Ah01}
\begin{equation}\label{scarf}
  V(x)=V_{\mathrm{re}}(x)+iV_{\mathrm{im}}(x)=V_0\,\mathrm{sech}^2(x)+iW_0\,\mathrm{sech}(x)\tanh(x).
\end{equation}
In the optical wave propagation, the real-valued external potential $V_{\mathrm{re}}(x)$ is responsible for the refractive index, and $V_{\mathrm{im}}(x)$
is usually used to describe the gain-and-loss distribution of the optical potential.
Here the real-valued parameters $V_0$ and $W_0$ can be employed to modulate the amplitudes of external potential, and gain-and-loss distribution, respectively.
Then the 1D NLS equation (\ref{NLS1d}) with $\mathcal{PT}$ Scarf-II potential (\ref{scarf}) admits the following exact soliton solution~\cite{pt1,pt2,pt3}
\begin{equation}\label{scarf1dsolu}
  u(x)=\sqrt{-\frac{2+V_0+W_0^2/9}{g}}\,\mathrm{sech}(x)\exp\left[-\frac{iW_0}{3}\arctan(\sinh(x))\right],
\end{equation}
where $\mu=-1$, $2+V_0+W_0^2/9>0$ in self-focusing case ($g=-1$) and $2+V_0+W_0^2/9<0$ in self-defocusing case ($g=1$).

Considering that the solution is a complex-valued function, in practical we set the network's output $\hat{u}(x)=p(x)+iq(x)$ and then separate Eq.~(\ref{NLS1d}) into its real and imaginary parts.
\begin{equation}\label{pq}
 \begin{array}{l}
     \displaystyle\mathcal{F}_p(x):= -\partial_{xx}p+V_{\mathrm{re}}p-V_{\mathrm{im}}(x)q+g(p^2+q^2)p-\mu p,\v\\
     \displaystyle\mathcal{F}_q(x):= -\partial_{xx}q+V_{\mathrm{re}}q+V_{\mathrm{im}}(x)p+g(p^2+q^2)q-\mu q.
 \end{array}
\end{equation}
Then the loss function $\mathcal{L}_2$ becomes
\begin{equation}\label{L2new}
  \mathcal{L}_2:=\frac{1}{N}\frac{\sum_{i=1}^{N}\left(|\mathcal{F}_p(x_i)|^2+|\mathcal{F}_q(x_i)|^2\right)}{\max_i\left(\sqrt{(p(x_i)^2+q(x_i)^2}\right)}.
\end{equation}

In self-focusing case ($g=-1$), we let $V_0=-1$ and $W_0=-1$.
Then based on IINN method, we set $\Omega=[-10,10]$ with $N=200$, and take the initial value as
\begin{equation}\label{scarfu0}
  u_0(x)=\mathrm{sech}(x)e^{ix}.
\end{equation}
After 2000 steps of iterations, with $\mathrm{NN}_1$ taking 5s and 25000 steps of iterations, with $\mathrm{NN}_2$ taking 322s, the relative $L_2$ errors $E_1$ of $u(x)$, $p(x)$ and $q(x)$, respectively, are 6.839215e-04, 9.031551e-04 and 1.579135e-03 compared to the exact solution (\ref{scarf1dsolu}). Figs.~\ref{f-scarf}(a1-a4) exhibit the comparison of real part, imaginary part and intensity $|u(x)|$  between the learned solutions and exact solutions as well as the loss-iteration diagram for $\mathrm{NN}_2$.

In self-defocusing case ($g=1$), we take $V_0=-3$ and $W_0=-1$.
Similarly, according to IINN method, we set $\Omega=[-10,10]$ with $N=200$, and take the initial value as
\begin{equation}\label{scarfu1}
  u_0(x)=\mathrm{sech}(x)e^{ix}.
\end{equation}
After 2000 steps of iterations, with $\mathrm{NN}_1$ taking 5s and 25000 steps of iterations, with $\mathrm{NN}_2$ taking 323s, the relative $L_2$ errors $E_1$ of $u(x)$, $p(x)$ and $q(x)$, respectively, are 5.084394e-04, 1.571929e-03 and 2.666748e-03 compared to the exact solution (\ref{scarf1dsolu}). Figs.~\ref{f-scarf}(b1-b4) exhibit the comparison of real part, imaginary part and intensity $|u(x)|$ between the learned solutions and exact solutions as well as the loss-iteration diagram for $\mathrm{NN}_2$.


\v \noindent {\bf Example 3.2} (The gap soliton of 1D saturable NLS equation with $\mathcal{PT}$-symmetric optical lattice).
When the nonlinear term is taken as the saturable nonlinearity,
\begin{equation}\label{sa}
  \mathcal{N}(x,|U|^2)U=\frac{g|U|^2U}{1+s|U|^2},
\end{equation}
and  $V(x)$ is taken as $\mathcal{PT}$-symmetric optical lattice (OL)
\begin{equation}
\label{1dol}
V(x)=V_0\cos(2x)+iW_0\sin(2x),
\end{equation}
the generalized NLS equation (\ref{U}) becomes
the 1D saturable NLS equation (SNLS) with
$\mathcal{PT}$-symmetric optical lattice
\bee
 iU_t-U_{xx}+V(x)U+\frac{g|U|^2U}{1+s|U|^2}=0,
\ene
where $s>0$ stands for the degree of saturable nonlinearity. Eq.~(\ref{lu}) is rewritten as
\begin{equation}\label{sNLS1d}
   L u =0, \quad  L= -\frac{1}{2}\partial_{xx} + V(x) + \frac{g|u|^2}{1+s|u|^2}-\mu,
\end{equation}

Specifically, we consider the self-defocusing case $(g=1)$ and take $s=0.3$, $V_0=3$, and $W_0=0.5$. For this case, the fundamental solitons can be found in the first gap at $\mu=-0.1$.

By the IINN method, the initial condition is taken as
\begin{equation}\label{snlsu0}
  u_0(x)=\mathrm{sech}(x)\cos(x)e^{ix},
\end{equation}
and the computational domain is set as $\Omega =[-8,8]$ with $N=800$. Considering the solution is complex, similar to the previous example, we write $\hat{u}(x)=p(x)+iq(x)$.
Then after 15000 steps of iterations, with $\mathrm{NN}_1$ taking 35s and 20000 steps of iterations, with $\mathrm{NN}_2$ taking 271s, the relative $L_2$ errors $E_1$ of $u(x)$, $p(x)$ and $q(x)$, respectively, are 4.008711e-04, 1.924903e-03 and 1.130601e-03 compared to the exact solution (numerically obtained with the same parameters as the previous example). Figs.~\ref{f-snls1d}(a1, a2, a3) exhibit the intensity diagram of real part, imaginary part and $|u(x)|$. The loss-iteration plot of $\mathrm{NN}_2$ is displayed in Fig.~\ref{f-snls1d}(a4).

\begin{figure*}[!t]
    \centering
  {\scalebox{0.85}[0.85]{\includegraphics{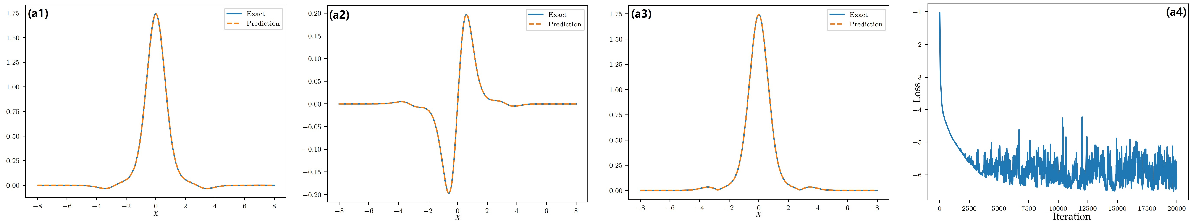}}}
  \hspace{-0.4in}\vspace{0.15in}
\vspace{0.1in}\caption{\small \rm  The complex solution $u(x)$ of 1D SNLS equation (\ref{sNLS1d}) with $\mathcal{PT}$-symmetric optical lattices (\ref{1dol}) in self-defocusing case. (a1, a2, a3) The real part, imaginary part and intensity diagrams of learned solution and exact one at $s=0.3$, $\mu=-0.1$, $V_0=3$, and $W_0=0.5$.
(a4) The loss-iteration plot of $\mathrm{NN}_2$.}
  \label{f-snls1d}
\end{figure*}


\v \noindent {\bf Example 3.3} (Soliton solutions of 1D  focusing-defocusing coupled nonlinear Schr\"{o}dinger equations).
Next, we consider the single-soliton solutions of 1D focusing-defocusing coupled nonlinear Schr\"{o}dinger (fdCNLS) equations given as follows
\begin{equation}\label{fdnls}
  \begin{split}
   & iU_{1t}+U_{1xx}+(|U_1|^2-|U_2|^2)U_1=0, \\
   & iU_{2t}+U_{2xx}+(|U_1|^2-|U_2|^2)U_2=0.
  \end{split}
\end{equation}
The equation admits the solitary waves in the form
\begin{equation}\label{fdnlsuv}
  \{U_1,U_2\}=\{u_1(x),u_2(x)\}e^{i\mu t}.
\end{equation}
Substituting them into Eqs.~(\ref{fdnls}) yields the equations for $u_1(x)$ and $u_2(x)$
\begin{equation}\label{fdnls1}
  \begin{split}
   \mathcal{F}_{u_1}:=&\, u_{1xx}+(|u_1|^2-|u_2|^2)u_1-\mu u_1=0, \\
   \mathcal{F}_{u_2}:=&\, u_{2xx}+(|u_1|^2-|u_2|^2)u_2-\mu u_2=0.
  \end{split}
\end{equation}
Then Eq.~(\ref{fdnls1}) [cf. Eq.~(\ref{L0})] is rewritten as
\begin{equation}\label{fdnlsL}
   \mathbf{L}_0 \mathbf{u} =0, \quad \mathbf{L}_0=
   \left(\begin{array}{cc}
 \partial_{xx} + (|u_1|^2-|u_2|^2)-\mu & 0 \\
  0 & \partial_{xx} + (|u_1|^2-|u_2|^2)-\mu \\
 \end{array} \right), \quad \mathbf{u}
 =\left(                                                       \begin{array}{c} u_1 \\                                                       u_2                                                        \end{array}                                                                  \right).
\end{equation}
The 1D fdCNLS equations (\ref{fdnls1}) admit the following exact solutions
\begin{equation}\label{fdnlssolu}
  (u_1,u_2)^\top=(A,B)^\top\mathrm{sech}(cx),
\end{equation}
where $\mu=c^2$ and $A^2-B^2-2c^2=0$.
Specifically, we take $c=2$ and $B=1$. For the case, the single-soliton solutions can be found at $\mu=4$.

Similar to the previous example, the loss function $\mathcal{L}_2$ can be written as
\begin{equation}\label{L2new1}
  \mathcal{L}_2:=\frac{1}{N}\frac{\sum_{i=1}^{N}\left(|\mathcal{F}_{u_1}(x_i)|^2+|\mathcal{F}_{u_2}(x_i)|^2\right)}{\max_i\left(\sqrt{(u_1(x_i)^2+u_2(x_i)^2}\right)}.
\end{equation}
By the IINN method, the initial condition is taken as
\begin{equation}\label{fdnlsu0}
  \{u_{10}(x), u_{20}(x)\}=\{2\mathrm{sech}(x), \mathrm{sech}(x)\},
\end{equation}
and the computational domain is set as $\Omega =[-15,15]$ with $N=500$.
Then after 5000 steps of iterations, with $\mathrm{NN}_1$ taking 11s and 20000 steps of iterations, with $\mathrm{NN}_2$ taking 251s, the relative $L_2$ errors $E_1$ of $u_1(x)$ and $u_2(x)$, respectively, are 1.569412e-03 and 2.026548e-03 compared to the exact solution. Figs.~\ref{f-fdnls}(a1, a2) display the comparison of $u_1(x)$ and $u_2(x)$ between
the learned solutions and exact solutions. The loss-iteration plot of $\mathrm{NN}_2$ is displayed in Fig.~\ref{f-fdnls}(a3).

\begin{figure*}[!t]
    \centering
  {\scalebox{0.81}[0.81]{\includegraphics{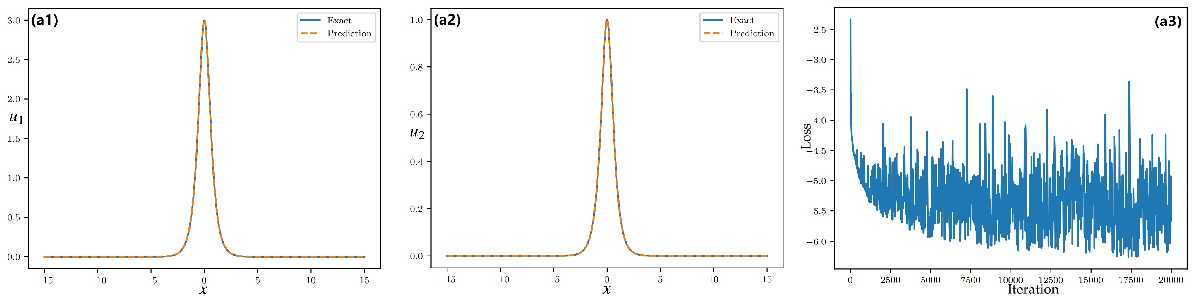}}}
  \hspace{-0.4in}\vspace{0.15in}
\vspace{0.1in}\caption{\small \rm  The single-soliton solutions $u_1(x)$ and $u_2(x)$ of 1D fdCNLS equations (\ref{fdnls1}). (a1, a2) The comparison of $u_1(x)$ and $u_2(x)$ between
the learned solutions and exact solutions at $A=3$, $B=1$, $c=2$, and $\mu=4$.
(a3) The loss-iteration plot of $\mathrm{NN}_2$.}
  \label{f-fdnls}
\end{figure*}


\begin{figure*}[!t]
    \centering
  {\scalebox{0.8}[0.8]{\includegraphics{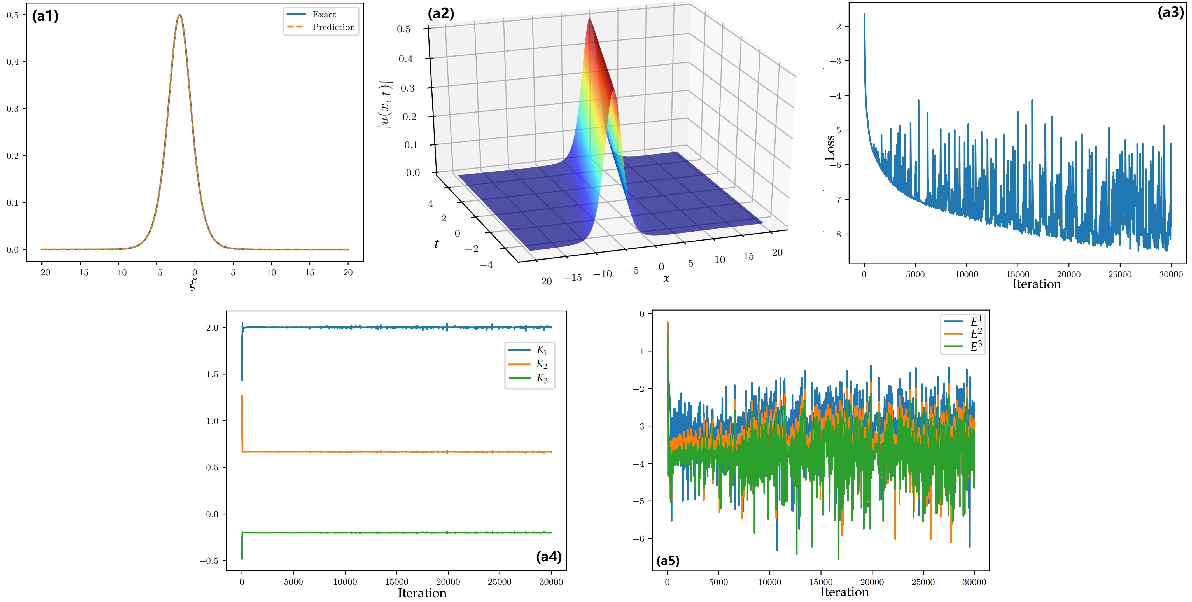}}}
  \hspace{-0.4in}\vspace{0.15in}
\vspace{0.1in}\caption{\small \rm The traveling wave solution $u(\xi)$ ($\xi=x-ct$) of KdV equation (\ref{kdv}). (a1) The learned solution and exact one at $c=1$ and $a=2$.
(a2) The 3D profile of the learned solution. (a3) The loss-iteration plot of $\mathrm{NN}_2$.
(a4) The conserved quantities $\int_{\mathbf{R}}\omega\dd x$ versus iteration, where $K_1$, $K_2$ and $K_3$ denote $\omega_1=U$, $\omega_2=U^2$ and $\omega_3=\frac{1}{2}U_x^2-U^3$, respectively.
(a5) The conserved quantity error versus iteration, where $E^i=\log_{10}|K_i-K_i^*|$ and $K_i^*$ is the true value of conserved quantity.}
  \label{f-kdv}
\end{figure*}

\v \noindent {\bf Example 3.4} (The solitary wave solution of KdV equation). The next example we consider is the KdV equation given as follows
\begin{equation}\label{kdv}
  U_t+6UU_x+U_{xxx}=0.
\end{equation}
Considering the traveling wave transform $\xi=x-ct$, then $U(x,t)=u(x-ct)=u(\xi)$ and one obtains
\begin{equation}\label{kdvu}
  -c\frac{\dd u}{\dd \xi}+6u\frac{\dd u}{\dd \xi}+\frac{\dd^3 u}{\dd \xi^3}=0,
\end{equation}
We can integrate this with respect to $\xi$ to obtain
\begin{equation}\label{kdvu1}
  -cu+3u^2+\frac{\dd^2 u}{\dd \xi^2}=A,
\end{equation}
where $A$ is a constant of integration. Therefore we consider the following nonlinear wave system
\begin{equation}\label{kdvL}
  Lu-A=0, \quad  L=\frac{\dd^2}{\dd \xi^2}+3u-c,
\end{equation}
When $A=0$, the solitary wave solution of the KdV equation can be found,
\begin{equation}\label{kdvsolux}
  u(\xi)=\frac{1}{2}c\,\mathrm{sech}^2\left[\frac{\sqrt{c}}{2}(\xi+a)\right],
\end{equation}
where $a$ is an arbitrary constant.

It should be noted that for Eq.~(\ref{kdvL}), there exist infinitely many solutions for given constant $c$.
If we use traditional PINNs method, we may not know which solution we will obtain.
Based on the IINN method, we take the initial value
\begin{equation}\label{kdvu0}
  u_0(\xi)=\mathrm{sech}^2(\xi+2),
\end{equation}
and consider $c=1$, $\Omega=[-20,20]$ with $N=500$. After 12000 steps iterations with 28s for $\mathrm{NN}_1$ and 30000 steps iterations with 235s for $\mathrm{NN}_2$, the relative $L_2$ error $E_1$=8.506870e-04 with exact solution (\ref{kdvsolux}) at $a=2$. Figs.~\ref{f-kdv}(a1, a3) displays the comparison between the learned solutions and exact
solutions at $c=1$ and $a=2$ as well as the loss-iteration diagram for $\mathrm{NN}_2$.
The 3D profile of the traveling wave solution $u(\xi)=u(x-ct)$ is shown in Fig.~\ref{f-kdv}(a2).
Furthermore, by changing the value of parameter $a$ in the initial condition $u_0$, we can obtain solutions at different positions.

Furthermore, since the integrability of KdV equation (\ref{kdv}), the three specific conservation laws are given as follows~\cite{Ablowitz-2}
\begin{equation}\label{cl1}
  U_t=(-U_{xx}-3U^2)_x,
\end{equation}
\begin{equation}\label{cl2}
  (U^2)_t=(-2UU_{xx}+U_x^2-4U^3)_x,
\end{equation}
\begin{equation}\label{cl3}
  (\frac{1}{2}U_x^2-U^3)_t=(-U_{x}U_{xxx}+\frac{1}{2}U_{xx}^2+3U^2U_{xx}-6UU_x^2+\frac{9}{2}U^4)_x.
\end{equation}
Similarly, these three conserved quantities \bee
K_1=\int_{\mathbf{R}}\omega\dd x=\int_{\mathbf{R}}U\dd x,\quad
K_2=\int_{\mathbf{R}}\omega\dd x=\int_{\mathbf{R}}U^2\dd x,\quad
K_3=\int_{\mathbf{R}}\omega\dd x=\int_{\mathbf{R}}\left(
\frac12U_x^2-U^3\right)\dd x
\ene versus iteration are displayed in Fig.~\ref{f-kdv}(a4).
And the variations of their error $E^i$, during the iteration are also displayed in Figs.~\ref{f-kdv}(a5), where $E^i=\log_{10}|K_i-K_i^*|$ and $K_i^*$ is the true value of conserved quantity.

On the other hand, in order to verify the importance of initial value selection, we provide numerical examples with different initial states that do not have the correct form or have different forms.
Firstly, we change the initial value $u_0(\xi)$ (\ref{kdvu0}) to another form
\begin{equation}\label{kdvu01}
  u_{01}(\xi)=\mathrm{sech}^2(\xi).
\end{equation}
Fixing the other parameters constant, we obtain the learned solution (see Fig.~\ref{f-kdvr} (a1)), after 12000 steps iterations for $\mathrm{NN}_1$ and 30000 steps iterations with for $\mathrm{NN}_2$.
It can be seen that the network still converges according to the loss-iteration diagram for $\mathrm{NN}_2$ (see Fig.~\ref{f-kdvr} (a2)). But the center of the solitary wave is at $x=0$, which is not what we want.
In particular, we take another initial value as
\begin{equation}\label{kdvu02}
  u_{02}(\xi)=\sin(\xi)\mathrm{sech}^2(\xi+2),
\end{equation}
which can be regarded as a modified one with varying amplitude ($\sin(\xi)$) of the previous initial value $u_0(\xi)$ given by Eq.~(\ref{kdvu0}). After 12000 steps iterations with for $\mathrm{NN}_1$ and 30000 steps iterations for $\mathrm{NN}_2$,
Fig.~\ref{f-kdvr}(a4) shows that the loss error for $\mathrm{NN}_2$ can only drop to around $10^{-3}$ orders of magnitude (in this case, the network is considered not convergent) such that the target solitary wave solution (see the solid line in Fig.~\ref{f-kdvr}(a3)) can not be obtained by using the initial value (\ref{kdvu02}), which generates the result (see the dashed line in Fig.~\ref{f-kdvr}(a3)). Therefore, the performances of IINN rely heavily on the suitable initial guess of the solitary wave solution.

\begin{figure*}[!t]
    \centering
  {\scalebox{0.89}[0.89]{\includegraphics{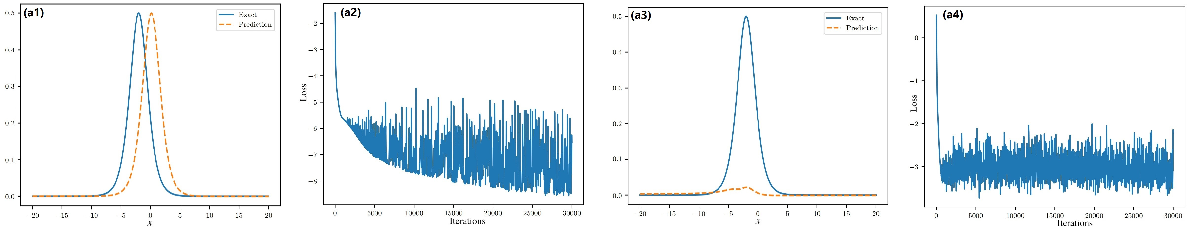}}}
  \hspace{-0.4in}\vspace{0.15in}
\vspace{0.1in}\caption{\small \rm The traveling wave solution $u(\xi)$ ($\xi=x-ct$) of KdV equation (\ref{kdv}). (a1) The learned solution and exact one at $c=1$ and $a=2$ by choosing initial value as $u_{01}(\xi)$ (\ref{kdvu01}). (a2) The loss-iteration plot of $\mathrm{NN}_2$.
(a3) The learned solution and exact one at $c=1$ and $a=2$ by choosing initial value as $u_{02}(\xi)$ (\ref{kdvu02}). (a2) The loss-iteration plot of $\mathrm{NN}_2$.}
  \label{f-kdvr}
\end{figure*}

\begin{remark}
  Here, Eq.~(\ref{kdvu}) is reduced to Eq.~(\ref{kdvu1}) by integrating. Then, we calculate Eq.~(\ref{kdvu1}) based on IINN method. We can calculate the Eq.~(\ref{kdvu}) directly. However, due to the automatic differential algorithm, the calculation time and error will increase. For example, for Eq.~(\ref{kdvu}), it takes twice as long to train $\mathrm{NN}_2$ as it does for Eq.~(\ref{kdvu1}). Therefore, for higher-order equations, order reduction is a good way to speed up the calculation.
\end{remark}


\v \noindent {\bf Example 3.5} (Ground state and vortex soliton of the 2D NLS equation with harmonic-oscillator (HO) trapping potential). The next example is the 2D self-focusing NLS equation with HO trapping potential
\bee\label{2d-nls}
 iU_t-\Delta_2 U+V(x,y)U-|U|^2U=0,
\ene
where $\Delta_2=\partial_x^2+\partial_y^2$, and $V(x,y)$ is the HO trapping potential in the form
\begin{equation}\label{HO2}
  V(x,y)=\frac{1}{2}(x^2+y^2).
\end{equation}

The stationary solution $U(\mathbf{x},t)=u(\mathbf{x})e^{i\mu t}$ makes Eq.~(\ref{2d-nls}) become
\begin{equation}\label{NLS2d}
   L u(\mathbf{x}) =0, \quad  \mathrm{where} \quad  L= -\Delta_2 + V(x,y) - |u|^2-\mu,
\end{equation}
where $\mathbf{x}=(x,y)$.

For ground state, we consider the computational domain $\Omega=[-5,5]\times [-5,5]$ with $N=20000$, and take the initial value
\begin{equation}\label{nls2du0}
  u_0(x,y)=e^{-0.5(x^2+y^2)}.
\end{equation}
Then the learned 2D ground state solution can be obtained at $\mu=0.5$, whose intensity diagram $|u(x,y)|$ and 3D profile are shown in Figs.~\ref{f-nls2d}(a1, a2), after 2000 steps of iterations with $\mathrm{NN}_1$ taking 7s and 20000 steps of iterations with $\mathrm{NN}_2$ taking 531s. The relative $L_2$ error $E_1$=4.624272e-04 compared to the exact solution $u$ (numerically obtained). The module of absolute error $|\hat{u}-u|$ is exhibited in Fig.~\ref{f-nls2d}(b1). The loss-iteration plot of $\mathrm{NN}_2$ for ground state is displayed in Fig.~\ref{f-nls2d}(b2).

\begin{figure*}[!t]
    \centering
  {\scalebox{0.85}[0.85]{\includegraphics{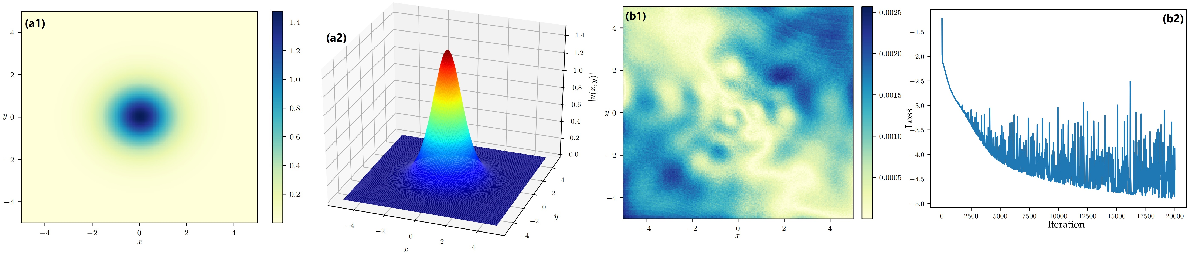}}}
  \hspace{-0.4in}\vspace{0.15in}
\caption{\small \rm The ground state solution $u(\mathbf{x})$ of 2D NLS equation (\ref{NLS2d}). (a1) The intensity diagram $|u(\mathbf{x})|$ of learned solution at $\mu=0.5$.
(a2) The 3D profile of the learned solution. (b1) The module of absolute error between the exact and learned solutions.
(b2) The loss-iteration plot of $\mathrm{NN}_2$.}
  \label{f-nls2d}
\end{figure*}

Furthermore, Eq.~(\ref{NLS2d}) simultaneously admits vortex soliton at $\mu=0.5$. By IINN method, the initial condition is taken as
\begin{equation}\label{nls2dvu0}
 u_0(r,\phi)=3re^{-0.5r^2}e^{i\phi},
\end{equation}
where $(r,\phi)$ is the polar coordinate of the $(x,y)$ plane (see, e.g., Ref.~\cite{boris}). The computational domain is set as $\Omega=[-5,5]\times [-5,5]$ with $N=20000$.
Because the vortex soliton solution is complex, similar to the previous example, we write $\hat{u}(x,y)=p(x,y)+iq(x,y)$.
Then after 10000 steps of iterations, with $\mathrm{NN}_1$ taking 36s and 20000 steps of iterations, with $\mathrm{NN}_2$ taking 992s, the relative $L_2$ errors $E_1$ of $u(x,y)$, $p(x,y)$ and $q(x,y)$, respectively, are 1.833355e-03, 6.223326e-03 and 6.214960e-03 compared to the exact solution (numerically obtained). Figs.~\ref{f-vortex}(a1, a2, a3) exhibit the intensity diagram of real part, imaginary part and $|u(x,y)|$. The 3D profile is shown in Fig.~\ref{f-vortex}(b1). The module of absolute error $|\hat{u}-u|$ is shown in Fig.~\ref{f-vortex}(b2).
And the loss-iteration plot of $\mathrm{NN}_2$ for vortex soliton is displayed in Fig.~\ref{f-vortex}(b3).

\begin{figure*}[!t]
    \centering
\vspace{0.1in}  {\scalebox{0.85}[0.85]{\includegraphics{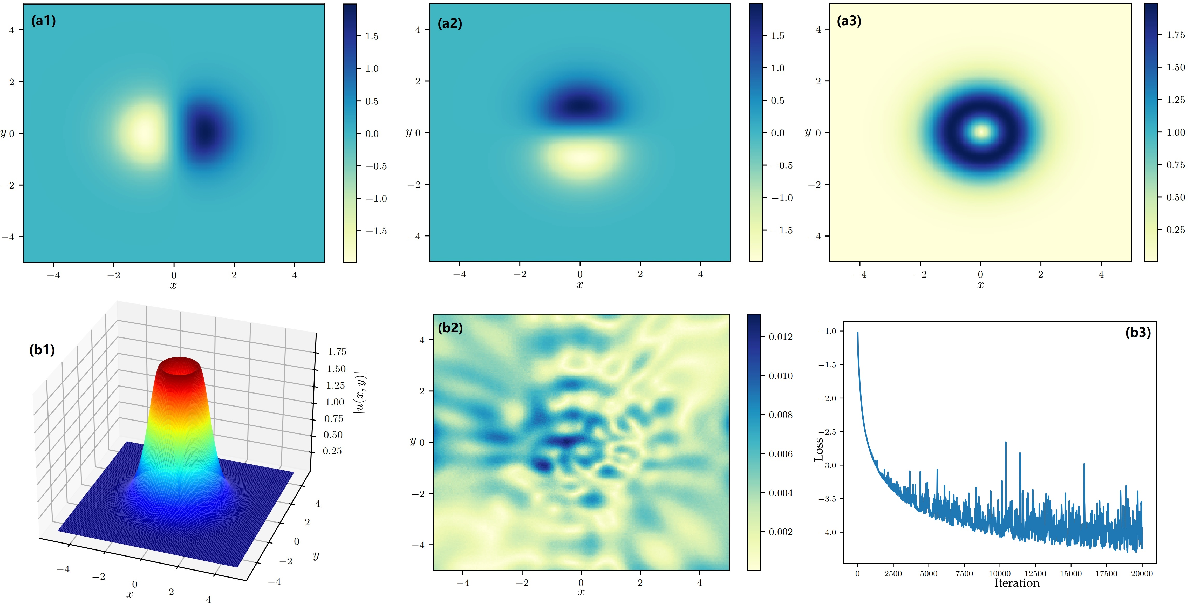}}}
  \hspace{-0.4in} \vspace{0.15in}
\vspace{0.1in}\caption{\small \rm The vortex soliton $u(\mathbf{x})$ of 2D NLS equation (\ref{NLS2d}). (a1, a2, a3) The real part, imaginary part and intensity  $|u(\mathbf{x})|$ diagrams of learned solution and exact one at $\mu=0.5$. (b1)  The 3D profile of the learned solution. (b2) The module of absolute error between the exact and learned solutions.
(b3) The loss-iteration plot of $\mathrm{NN}_2$.}
  \label{f-vortex}
\end{figure*}


\v \noindent {\bf Example 3.6} (Gap soliton of the 2D NLS equation with periodic potential).
The example we consider is the computation of gap solitons in the 2D defocusing NLS equation (\ref{2d-nls}) with the optical lattice potential
\bee\label{2d-nls-2}
 iU_t-\Delta_2 U+V(x,y)U+|U|^2U=0,
\ene
where $\Delta_2=\partial_x^2+\partial_y^2$, and the optical lattice potential is
\begin{equation}\label{OL}
  V(x,y)=V_0\left(\sin^2x+\sin^2y\right),\qquad V_0\in\mathbb{R}.
\end{equation}

The stationary solution $U(x,y,t)=u(x,y)e^{i\mu t}$ makes Eq.~(\ref{2d-nls-2}) become
\begin{equation}\label{NLS2d2}
   L u(x,y) =0, \quad  \mathrm{where} \quad  L= -\Delta_2 + V(x,y) + |u|^2-\mu.
\end{equation}
Eq.~(\ref{NLS2d2}) with periodic potential (\ref{OL}) at $V_0=6$ admits soliton solutions in the first bandgap. By the IINN method, the initial condition is taken as
\begin{equation}
\label{olu0}
u_0(x,y)=\mathrm{sech}\Big(\sqrt{x^2+y^2}\Big)\cos(x)\cos(y).
\end{equation}
And the computational domain is set as $\Omega=[-10,10]\times[-10,10]$ with $N = 20000$. Then the learned gap soliton can be found at $\mu =5$, whose intensity diagram $|u(x, y)|$ and 3D profile
are displayed in Figs.~\ref{f-ol2d}(a1, a2), after 20000 steps of iterations with $\mathrm{NN}_1$ taking 93s and 40000 steps of iterations with $\mathrm{NN}_2$ taking 1739s. The relative $L_2$ error $E_1$=4.674467e-03 compared to the exact solution $u$ (numerically obtained).
And the module of absolute error is exhibited in Fig.~\ref{f-ol2d}(b1). The loss-iteration plot of $\mathrm{NN}_2$ for gap soliton is displayed in Fig.~\ref{f-ol2d}(b2).

\begin{figure*}[!t]
    \centering
  {\scalebox{0.85}[0.85]{\includegraphics{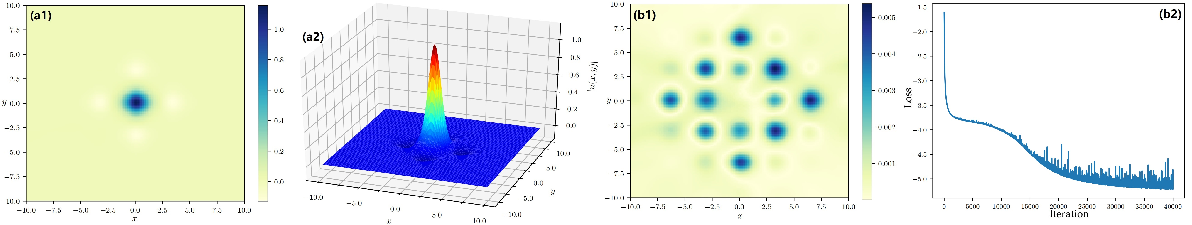}}}
  \hspace{-0.4in}\vspace{0.15in}
\vspace{0.1in}\caption{\small \rm The gap soliton $u(\mathbf{x})$ of 2D NLS equation (\ref{NLS2d}) with periodic potential (\ref{OL}). (a1) The intensity diagram $|u(\mathbf{x})|$ of learned solution
at  $\mu = 5$.
(a2) The 3D profile of the learned solution. (b1) The module of absolute error between the exact and learned solutions.
(b2) The loss-iteration plot of $\mathrm{NN}_2$.}
  \label{f-ol2d}
\end{figure*}

\begin{figure*}[!t]
    \centering
  {\scalebox{0.85}[0.85]{\includegraphics{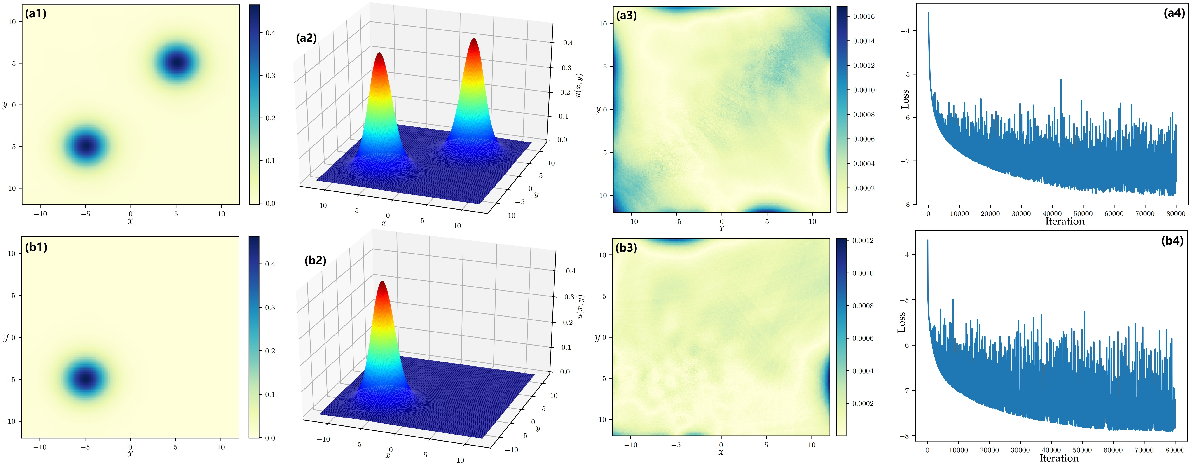}}}
  \hspace{-0.4in}\vspace{0.15in}
\vspace{0.1in}\caption{\small \rm The quantum droplets $u(\mathbf{x})$ of 2D amended GP equation with LHY correction (\ref{gp2}) (a1, b1) The intensity diagrams of 2D QDs in different branches at $\mu=-0.5$.
(a2, b2) The 3D profile of the learned solution in different branches. (a3, b3) The module of absolute error between the exact and learned solutions. (a4, b4) The loss-iteration plot of $\mathrm{NN}_2$.}
  \label{f-4well}
\end{figure*}

\v \noindent {\bf Example 3.7} (Quantum droplets of the 2D amended GP equation with LHY correction and multi-well potential).
In the next example, we calculate the symmetry breaking bifurcations of 2D quantum droplets (QDs) for the amended Gross-Pitaevskii equation with Lee-Huang-Yang (LHY) corrections and a Gaussian
quadruple-well potential \cite{2Dqd}. Here $\mathcal{N}(x,|U|^2)U$ is replaced by $2\ln(2|U|^2)|U|^2U$, and we have
\begin{equation}\label{GP}
   iU_t=\left[-\dfrac{1}{2}\Delta_2+V({\bf r})+2\ln(2|U|^2)|U|^2\right]U,
\end{equation}
where $\Delta_2=\partial_x^2+\partial_y^2$, and the 2D Gaussian quadruple-well potential is taken in the form
\begin{equation}\label{potential}
  V({\bf r})= V_0\sum_{j=1}^4\exp\left[-k|{\bf r}-{\bf r}_j|^2 \right],\quad V_0<0,\quad k>0,
\end{equation}
where ${\bf r}=(x,y)$, ${\bf r}_j=(\pm x_0,\pm y_0)$, $j=1,2,3,4$ control the locations of these four potential wells, and $|V_0|$ and $k$ regulate the depths and widths of potential wells, respectively.
We let $V_0 =-0.5$ and $k = 0.1$ in the following discussion. Here, $(x_0, y_0) = (5, 5)$ allows the four potential wells to be fully separated.

Analogously, we set $U({\bf r}, t)=u({\bf r})e^{-i\mu t}$, where $\mu$ stands for the chemical potential. Substituting the solution into Eq.~(\ref{GP}) yields the following nonlinear stationary equation
\begin{equation}\label{gp2}
  L u =0, \quad   L=-\frac{1}{2}\Delta_2+2\ln(2|u|^2)|u|^2+V({\bf r})-\mu.
\end{equation}
According to Ref.~\cite{2Dqd}, we know that there exist twelve different real solution branches and one complex solution branches. In the following, we calculate two of these branches for the same equation (\ref{gp2}) at $\mu=-0.5$.

In branch B1 (see the notation in Ref.~\cite{2Dqd}), we consider the computational domain $\Omega = [-12, 12]\times[-12, 12]$ with $N = 20000$, and take the initial value as
\begin{equation}\label{gpu01}
  u_0({\bf r})=0.3\left[e^{-0.1|{\bf r}-{\bf r}_1|^2}+e^{-0.1|{\bf r}-{\bf r}_3|^2}\right],\quad {\bf r}_1=(5,5),\quad {\bf r}_3=(-5,-5).
\end{equation}
Then the learned 2D QDs in branch B1 can be obtained at $\mu=-0.5$, whose intensity diagram $|u(x,y)|$ and 3D profile are shown in Figs.~\ref{f-4well}(a1, a2), after 15000 steps of iterations with $\mathrm{NN}_1$ taking 75s and 80000 steps of iterations with $\mathrm{NN}_2$ taking 3723s. The relative $L_2$ error $E_1$=4.291031e-03 compared to the exact solution (numerically obtained). The module of absolute error $|\hat{u}-u|$ is exhibited in Fig.~\ref{f-4well}(a3). The loss-iteration plot of $\mathrm{NN}_2$ for QDs in branch B1 is displayed in Fig.~\ref{f-4well}(a4).

In branch A1 (see the notation in Ref.~\cite{2Dqd}), we take the initial value as
\begin{equation}\label{gpu02}
  u_0({\bf r})=0.46e^{-0.1|{\bf r}-{\bf r}_3|^2},\quad {\bf r}_3=(-5,-5).
\end{equation}
After 10000 steps of iterations, with $\mathrm{NN}_1$ taking 55s and 80000 steps of iterations, with $\mathrm{NN}_2$ taking 3780s, the relative $L_2$ errors $E_1$=2.851064e-03 compared to the exact solution (numerically obtained). Figs.~\ref{f-4well}(b1, b2) exhibit the intensity diagram $|u(x,y)|$ and its 3D profile. The module of absolute error $|\hat{u}-u|$ is shown in Fig.~\ref{f-4well}(b3). The loss-iteration plot of $\mathrm{NN}_2$ for QDs in branch A1 is displayed in Fig.~\ref{f-4well}(b4).


\v \noindent {\bf Example 3.8} (Solitary-wave solution of Kadomtsev-Petviashvili equation).
Next we consider the $(2+1)$-dimensional KP equation with higher-order dispersion term
\begin{equation}\label{KP}
  (U_t+6UU_x+U_{xxx})_x+\alpha U_{yy}=0,\qquad \alpha\in\mathbb{R}.
\end{equation}
Similarly, we consider the traveling wave transform $\xi=x-ct$, then $U(x,y,t)=u(x-ct,y)=u(\xi,y)$ and one obtains
\begin{equation}\label{kpu}
  (-cu_{\xi}+6uu_{\xi}+u_{\xi\xi\xi})_{\xi}+\alpha u_{yy}=0.
\end{equation}
For Eq.~(\ref{kpu}), we can find the specific solitary wave solution as follows
\begin{equation}\label{kpsolu}
  u(\xi,y)=\frac{1}{2}(\alpha-c)\,\mathrm{sech}^2\left[\frac{\sqrt{\alpha-c}}{2}(\xi+y)\right].
\end{equation}

Then based on IINN method, we take the initial value
\begin{equation}\label{kpu0}
  u_0(\xi,y)=\mathrm{sech}^2(\xi+y),
\end{equation}
and consider $\alpha=2$, $c=1$ and $\Omega=[-5,5]\times[-5,5]$ with $N=20000$. After 10000 steps iterations with 33s for $\mathrm{NN}_1$ and 50000 steps iterations with 5537s for $\mathrm{NN}_2$, the relative $L_2$ error $E_1$=9.649114e-04 with exact solution (\ref{kpsolu}) at $\alpha=2$ and $c=1$. The intensity diagram $|u(\xi,y)|$ and its 3D profile are shown in Figs.~\ref{f-kp}(a1, a2). The module of absolute error $|\hat{u}-u|$ is exhibited in Fig.~\ref{f-kp}(b1). The loss-iteration plot of $\mathrm{NN}_2$ for traveling wave solution is displayed in Fig.~\ref{f-kp}(b2).

\begin{figure*}[!t]
    \centering
  {\scalebox{0.85}[0.85]{\includegraphics{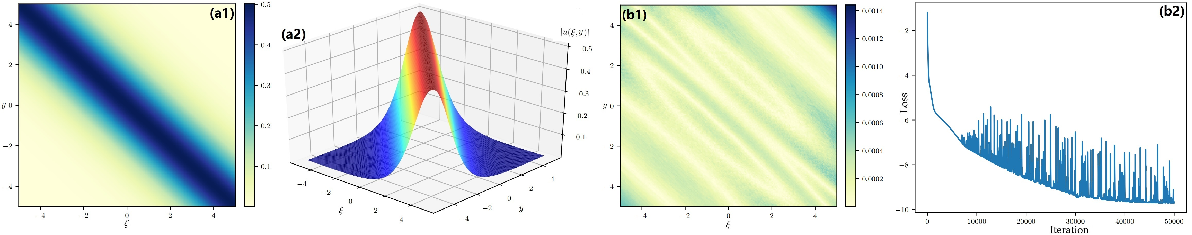}}}
  \hspace{-0.4in}\vspace{0.15in}
\vspace{0.1in}\caption{\small \rm The traveling wave solution $u(\xi,y)$ $(\xi = x-ct)$ of KP equation (\ref{KP}). (a1) The intensity diagram $|u(\xi,y)|$ of learned solution at $\alpha = 2$ and $c = 1$.
(a2) The 3D profile of the learned solution.  (b1) The module of absolute error between the exact and learned solutions. (b2) The loss-iteration plot of $\mathrm{NN}_2$.}
  \label{f-kp}
\end{figure*}
\begin{figure*}[!h]
    \centering
  {\scalebox{0.8}[0.8]{\includegraphics{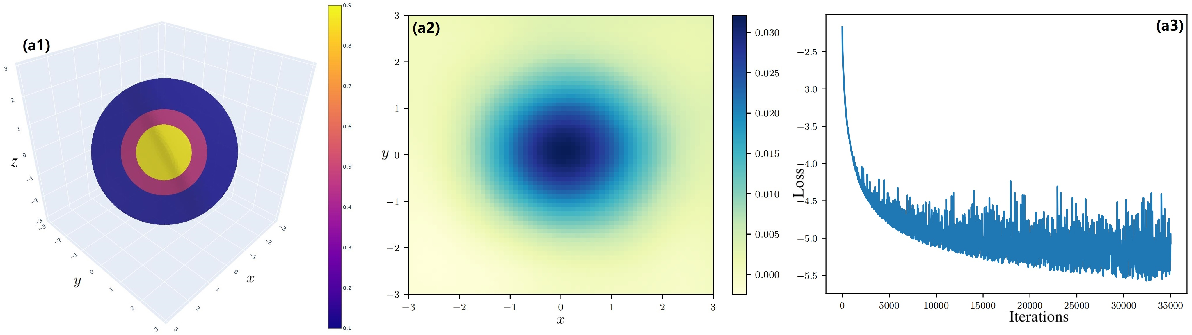}}}
  \hspace{-0.4in}\vspace{0.15in}
\vspace{0.1in}\caption{\small \rm The optical bullet (soliton solution) $u(\mathbf{x})$ of 3D NLS equation (\ref{NLS3d}). (a1) Isosurface of learned soliton at values $0.1$, $0.5$, $0.9$ at $\mu=1.5$. (a2) The 2D profile of the learned solution $u(\mathbf{x})$ at $z=0$.
(a3) The loss-iteration plot of $\mathrm{NN}_2$.}
  \label{f-nls3d}
\end{figure*}
It should be noted that for higher-order and higher-dimensional equations, the time required to train $\mathrm{NN}_2$ increases considerably.


\v \noindent {\bf Example 3.9} (Optical bullets of 3D NLS equation with HO trapping potential).
In the last example, we consider the 3D focusing NLS equation with HO trapping potential
\bee\label{2d-nls-3}
 iU_t-\Delta_3 U+V(x,y,z)U-|U|^2U=0,
\ene
where $\Delta_3=\partial_x^2+\partial_y^2+\partial_z^2$, and the HO trapping potential is taken in the form
\begin{equation}\label{HO3}
  V(x,y,z)=\frac{1}{2}(x^2+y^2+z^2).
\end{equation}

The stationary solution $U(x,y,z,t)=u(x,y,z)e^{i\mu t}$ makes Eq.~(\ref{2d-nls-3}) become
\begin{equation}\label{NLS3d}
   L u(\mathbf{x}) =0, \quad   L= -\Delta_3 + V(x,y,z) - |u|^2-\mu,
\end{equation}
where $\mathbf{x}=(x,y,z)$.

The computational domain is consider $\Omega=[-3,3]\times [-3,3]\times [-3,3]$ with $N=40000$, and the initial value is taken as
\begin{equation}\label{nls3du0}
  u_0(x,y,z)=e^{-0.5(x^2+y^2+z^2)}.
\end{equation}
Then the learned 3D optical bullet (soliton solution) can be obtained at $\mu=1.5$, after 10000 steps of iterations with $\mathrm{NN}_1$ taking 49s and 35000 steps of iterations with $\mathrm{NN}_2$ taking 1978s.
And the relative $L_2$ error $E_1$=8.551697e-03 compared to the exact solution (numerically obtained).
Fig.~\ref{f-nls3d}(a1) displays the isosurface of learned soliton at values $0.1$, $0.5$ and $0.9$.
And the 2D profile of the learned solution at $z = 0$ is exhibited in Fig.~\ref{f-nls3d}(a2). The loss-iteration plot of $\mathrm{NN}_2$ is displayed in Fig.~\ref{f-nls3d}(a3).
It should be noted that due to the automatic differentiation algorithm, the memory required by applying machine learning to solve high-dimensional systems is much less than that required by numerical methods.

Table \ref{table1} shows all the examples we considered in the following text, including the equations, the desired solitary wave solutions, the given initial values, the number of iterations, the number of training points, and the relative $L_2$ error.
\begin{table}[!ht]
\vspace{-0.15in}
\begin{center}\footnotesize
\caption{\small \rm The tested some examples and data via the IINN method.}\label{table1}
\vspace{-0.05in}
\begin{tabular}{ c|c|cccccc}
\hline
\hline\\[-2ex]
Equation & Potential    & Solution       & Initial value          & Step ($\mathrm{NN}_1$)      & Step ($\mathrm{NN}_2$)    & $N$   & $E_1$   \\[0.5ex]
\hline\\[-3ex]

\multirow{4}{*}{1D NLS}
\rule{0pt}{13pt}            &\multirow{2}{*}{/}                   & Bright soliton    & $u_0=\mathrm{sech}(x)$  & 10000 & 25000 & 500 & 1.63e-03\\[0.5ex]
\cline{3-8}\\[-3ex]
\rule{0pt}{13pt}            &                                     & Dark soliton    & $u_0=\tanh(x)$       &10000 & 14000 & 500 & 3.40e-04 \\[0.5ex]
\cline{2-8}\\[-3ex]

\rule{0pt}{13pt}      &\multirow{2}{*}{HG}       & Ground state    & $u_0=\exp(-x^2)$          & 5000 & 30000 & 200 & 2.51e-04\\[0.5ex]
\cline{3-8}\\[-3ex]
\rule{0pt}{13pt}            &                    & Dipole mode    & $u_0=4x\exp(-x^2/2)$       &10000 & 25000 & 200 & 4.66e-04 \\[0.5ex]
\cline{2-8}\\[-3ex]

\rule{0pt}{13pt}            &\multirow{2}{*}{$\mathcal{PT}$ Scarf-II}  &
Soliton (focusing)    & $u_0=\mathrm{sech}(x)\exp(ix)$ & 2000 & 25000 & 200 & 6.83e-04\\[0.5ex]
\cline{3-8}\\[-3ex]
\rule{0pt}{13pt}            &                    & Soliton (defocusing)   & $u_0=\mathrm{sech}(x)\exp(ix)$       &2000 & 25000 & 200 & 5.08e-04\\[0.5ex]
\hline\\[-3ex]

\rule{0pt}{13pt}   1D SNLS   &$\mathcal{PT}$ OL      & Gap soliton   & $u_0=\mathrm{sech}(x)\cos(x)\exp(ix)$  & 15000 & 20000 & 800 & 4.00e-04\\[0.5ex]

\hline\\[-3ex]

\rule{0pt}{13pt}   1D fdCNLS   &/        & Single-soliton   & $\{u_{10}, u_{20}\}=\{2,1\}\mathrm{sech}(x)$  & 5000 & 20000 & 500 & 1.56e-03\\[0.5ex]

\hline\\[-3ex]

\rule{0pt}{13pt}  1D KdV   &/        & Solitary wave   & $u_0=\mathrm{sech}^2(\xi+a)$  & 12000 & 30000 & 500 & 8.50e-04\\[0.5ex]

\hline\\[-3ex]

\multirow{3}{*}{2D NLS}
\rule{0pt}{13pt}     &\multirow{2}{*}{HO}    & Ground state   & $u_0=e^{-0.5(x^2+y^2)}$          & 2000& 20000& 20000 & 4.62e-04\\[0.5ex]
\cline{3-8}\\[-3ex]
\rule{0pt}{13pt}           &                         & Vortex soliton    & $u_0=3re^{-0.5r^2}e^{i\phi}$  & 10000& 20000& 20000 & 1.83e-03\\[0.5ex]
\cline{2-8}\\[-3ex]

\rule{0pt}{13pt}            &Periodic  & Gap soliton    & $u_0=\mathrm{sech}(\sqrt{x^2+y^2})\cos(x)\cos(y)$ & 20000 & 40000 & 20000 & 6.47e-03\\[0.5ex]
\hline\\[-3ex]

\multirow{2}{*}{2D GP}
\rule{0pt}{13pt}     &\multirow{2}{*}{Quadruple-well}    & Branch B1   & $u_0=0.3(e^{-0.1|{\bf r}-{\bf r}_1|^2}+e^{-0.1|{\bf r}-{\bf r}_3|^2})$      & 15000& 80000& 20000 & 4.29e-03\\[0.5ex]
\cline{3-8}\\[-3ex]
\rule{0pt}{13pt}                  &                                     & Branch A1    & $u_0=0.46e^{-0.1|{\bf r}-{\bf r}_3|^2}$  & 10000& 80000& 20000 & 2.85e-03\\[0.5ex]
\hline\\[-3ex]

\rule{0pt}{13pt}  2D KP   &/        & Solitary wave   & $u_0=\mathrm{sech}^2(\xi+y)$  & 10000 & 50000 & 20000 & 8.35e-04\\[0.5ex]
\hline\\[-3ex]

\rule{0pt}{13pt}   3D NLS   &HO        & Bullet   & $u_0=e^{-0.5(x^2+y^2+z^2)}$  & 10000 & 35000 & 40000 & 8.55e-03\\[0.5ex]
\hline
\hline
\end{tabular}
\end{center}
\end{table}

\begin{remark}
In summary, the key to the success of IINN is the choice of initial value $u_0$ as it determines the type of solution we ultimately obtain. Theorem 2 claims that the initial state $u_0$ is close enough to the real solution $u^*$. In fact, since we aim to compute solitary wave solutions with zero boundary conditions, this condition can be relaxed appropriately.
For example, in Example 3.4 for the KdV equation, the learned solutions can also achieve the same accuracy by taking another similar initial value as $u_0(\xi)=e^{-(\xi+2)^2}$. As described in Remark 2, based on the characteristics of the system and our understanding of the system, we can estimate the initial value using physical background knowledge or past experience.
If we know the form of the exact solution, we can give suitable initial values, such as Example 3.1 [Case 1, Case 2, Case 3 (ground state), Case 4], Example 3.3, Example 3.4, and Example 3.8. In the case that one does not know the form of exact solution, the initial conditions can be estimated according to the characteristics of the system.
For instance, for Examples 3.2 and 3.6 with periodic potential, we know the solitons originate from the Bloch-band edges. Therefore, we consider the cosine function term in the initial value.
On the other hand, the initial conditions can be obtained by computing the spectra and eigenmodes in the linear regime.
For example 3.5 and 3.9 with harmonic-oscillator trapping potential, we know the ground state in the linear regime in the form of $Ae^{-r^2}$, where $r$ is the radius in polar coordinates.
We can adjust the previous coefficient $A$ to make $|Lu_0|$ small enough. Furthermore, we exhibit an example to demonstrate the feasibility of this approach [see Case 3 (dipole mode) in Example 3.1].
\end{remark} 

\begin{figure*}[!t]
    \centering
  {\scalebox{0.57}[0.57]{\includegraphics{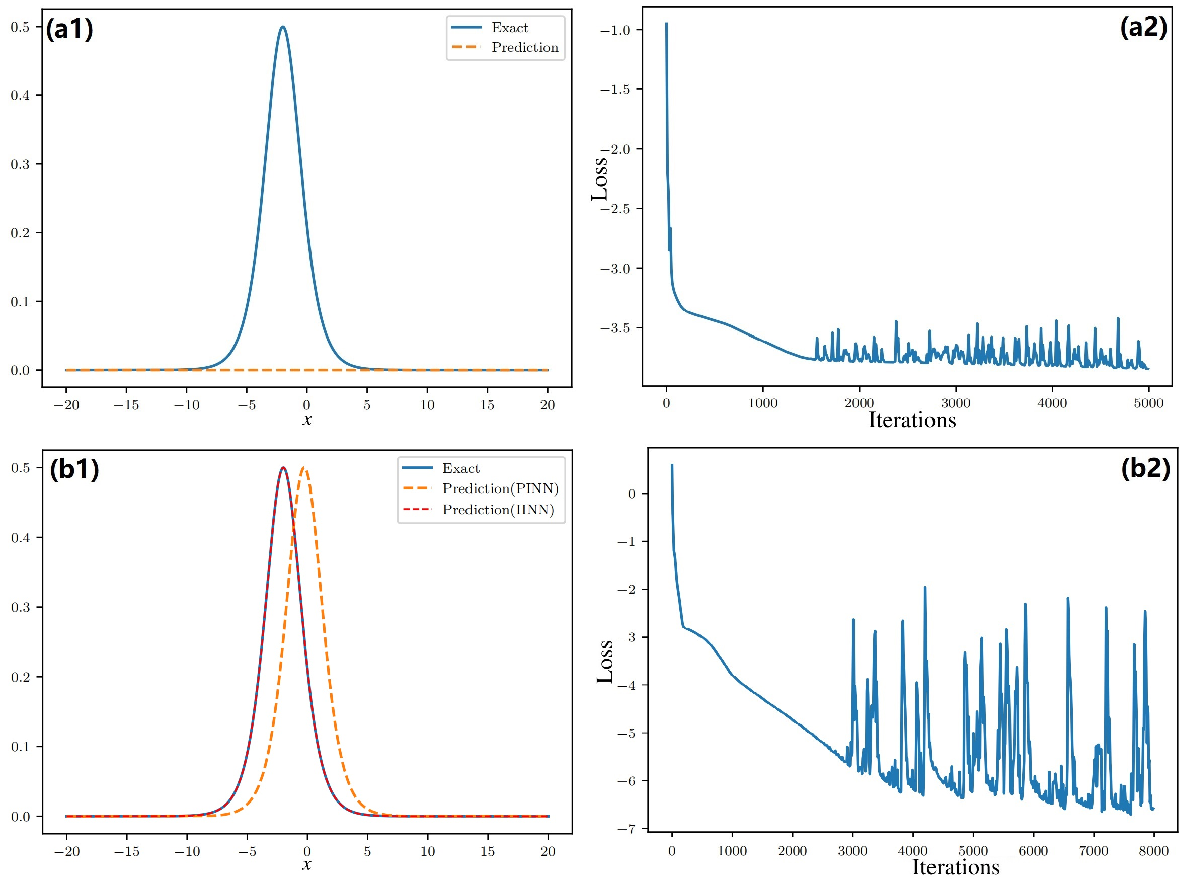}}}
\vspace{0.1in}\caption{\small \rm (a1, a2) The comparison of learned and exact solutions (traveling solitary wave) by classic PINNs method for KdV equation (\ref{kdv}) (Ref. Eq.~(\ref{kdvu1})), and loss-iteration plot.
(b1, b2) The comparison of learned and exact solutions by PINNs with randomly initialized parameters and loss function $\mathcal{L}_3$ given by Eq.~(\ref{Loss3}) for KdV equation, and loss-iteration plot.
}
  \label{f-pinn}
\end{figure*}

\section{Comparison between IINN method and traditional method}
In this section, we give the limitations of the PINNs method by comparison and present the advantages existing in IINN method compared to traditional numerical methods.

\subsection{Limitations of PINNs method}
Considering that the solitary wave for equation (\ref{L0}) is not unique, especially the equation has trivial solution, if we directly apply the PINNs method to the calculation of solitary wave, then we will almost certainly get the trivial solution. For example, for the KdV equation (\ref{kdv}) (Ref. Eq.~(\ref{kdvu1})) considered in the previous section, the network will converge to the trivial solution $u=0$ eventually with loss function $\mathcal{L}_0$ given by Eq.~(\ref{Loss0}) (see Figs.~\ref{f-pinn}(a1, a2)). Therefore, it is almost impossible to directly apply PINNs method to solve solitary waves unless additional information is given in the interior of the region.

Furthermore, if we replace the PDE residual term in the loss function $\mathcal{L}_0$ given by Eq.~(\ref{Loss0}) in PINNs with $\mathcal{L}_2$ given by Eq.~(\ref{Loss2}), that is
\begin{equation}\label{Loss3}
 \mathcal{L}_3:=\d\frac{1}{N_f}\frac{\sum_{\ell=1}^{N_f}|\mathbf{L}\hat{\mathbf{u}}(\mathbf{x}_f^{\ell})|^2}{\max(|\hat{\mathbf{u}}(\mathbf{x}_f^{\ell})|)}+\frac{1}{N_b}\sum_{\ell=1}^{N_b}|\hat{\mathbf{u}}(\mathbf{x}_b^{\ell})|^2,
\end{equation}
it may be that the network will converge to a non-trivial solution, but the solution may not be that we need. Similarly, we consider the KdV equation by the presented IINN method. After 8000 steps of iterations with $\mathrm{NN}_2$ with randomly initialized parameters,
although a solitary wave solution is obtained, it is not the desired one (see Fig.~\ref{f-pinn}(b1)).
If we use IINN method, we can obtain the solitary wave solution centered at any position (see the red dashed line in Fig.~\ref{f-pinn}(b1)).
This is because the equation has infinitely many solutions. And we do not know which one the network with randomly initialized parameters eventually converges to.
According to loss-iteration plot (see Fig.~\ref{f-pinn}(b2)), it can be found that the network has converged.

\begin{figure*}[!t]
    \centering
  {\scalebox{0.75}[0.75]{\includegraphics{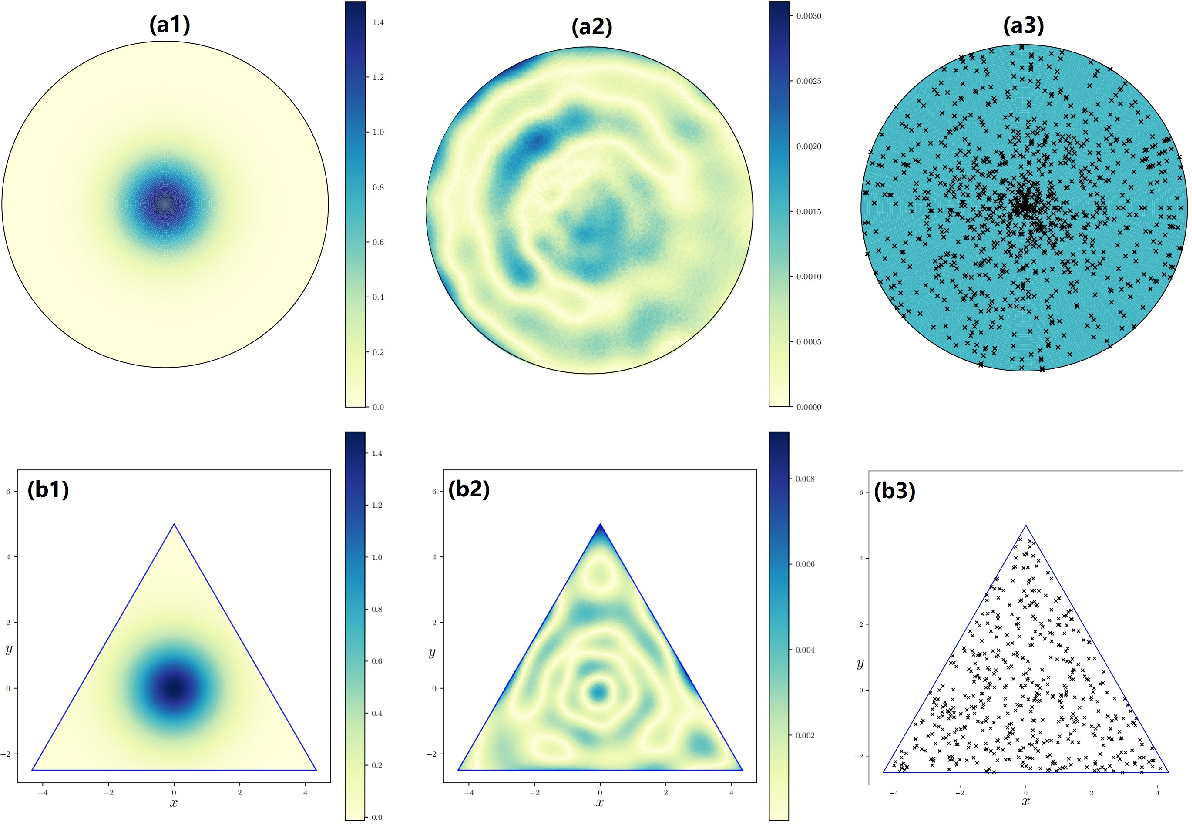}}}
  \hspace{-0.4in}\vspace{0.15in}
\vspace{0.1in}\caption{\small \rm The ground state solution $u(\mathbf{x})$ of 2D NLS equation on the disk. (a1) The intensity diagram $|u(\mathbf{x})|$ of learned solution at $\mu = 0.5$.
(a2) The module of absolute error between the exact and learned solutions. (a3) The randomly selected points on the disk.
The ground state solution $u(\mathbf{x})$ of 2D NLS equation on the equilateral triangle region. (b1) The intensity diagram $|u(\mathbf{x})|$ of learned solution at $\mu = 0.5$.
(a2) The module of absolute error between the exact and learned solutions. (a3) The randomly selected points on the equilateral triangle region.}
  \label{f-NLS2Dp}
\end{figure*}

\subsection{Advantages over traditional numerical methods}

IINN method has many advantages over traditional numerical methods (e.g, Ref.~\cite{spectral}). In general, in traditional numerical methods, differentiation is approximated by difference, which requires the domain to be meshed, and the error depends on the mesh size. Therefore, it is difficult to calculate the difference when dealing with complex region problems. However, due to the automatic differentiation algorithm, we can easily deal with derivatives when applying IINN method.
For example, when computing the ground state solution of 2D NLS equation with HO potential, we only need to consider on the disk, and the rest of the region is useless.
Therefore, we consider $\Omega=\left\{\mathbf{x}|d(\mathbf{x},0)\leq5\right\}$ with $N=1000$. With the same initial conditions as before, we can obtain the ground state solution through IINN method, whose intensity diagram on the disk is shown in Fig.~\ref{f-NLS2Dp}(a1). The relative $L_2$ error $E_1$=2.340686e-03 compared to the exact solution.
The module of absolute error is exhibited in Fig.~\ref{f-NLS2Dp}(a2).  The randomly selected points on the disk is shown in Fig.~\ref{f-NLS2Dp}(a3).
We can see that instead of 20000 training points, it now takes only 1000 training points to achieve the same accuracy.
More especially, we choose the equilateral triangle region with its center at the origin and side length $5\sqrt{3}$. Similarly, we compute the ground state solution of 2D NLS equation with HO potential. With the same initial conditions, we can obtain the ground state solution through IINN method with fewer training points $N=500$, whose intensity diagram is shown in Fig.~\ref{f-NLS2Dp}(b1). The relative $L_2$ error $E_1$=3.572674e-03 compared to the exact solution. The module of absolute error is exhibited in Fig.~\ref{f-NLS2Dp}(b2).  The randomly selected points on the triangle region is shown in Fig.~\ref{f-NLS2Dp}(b3).

On the other hand, when dealing with high-dimensional problems, IINN method has great advantages. For traditional numerical methods, the required memory often increases exponentially with the increase of dimension. However, for IINN method, the change of the dimension has little effect on the memory.
For instance, computing the 3D NLS equation requires only twice as much training points as computing the 2D NLS equation in the previous section.

\section{Summary}
We have proposed the initial value iterative neural network (IINN) algorithm for solitary wave computations.
IINN method combines the ideas of traditional numerical iterative methods and the principles of physics-informed neural networks (PINNs), which consists of two subnetworks. One subnetwork is utilized to fit the given initial value condition, while the other subnetwork incorporates physical information and continues training based on the first network.
Notably, the IINN approach does not require any data information including boundary conditions, except the given initial value. Furthermore, we provide corresponding theoretical guarantees to demonstrate the effectiveness of our method.

We apply the proposed method to compute both the ground states and excited states in a large number of physical systems, such as the one-dimensional NLS equation (with and without potentials), the one-dimensional NLS equation with saturable nonlinearity and $\PT$-symmetric optical lattices, the one-dimensional coupled focusing-defocusing NLS equations,
the KdV equation, the two-dimensional NLS equation, the two-dimensional amended GP equation, the (2+1)-dimensional KP equation, and the three-dimensional NLS equation, which demonstrate the effectiveness of our method. Finally, by comparing with traditional methods, we show the advantages of the IINN approach.

On the other hand, we should note that although the corresponding theoretical guarantees are given, the risk of algorithm divergence still exists.
This is because the choice of initial value $\mathbf{u}_0$ is crucial as it determines the type of solution we ultimately obtain. If the initial value is far from the exact solution, then our method may fail. Furthermore, the accuracy of our method may be closely related to the optimization algorithm. We can use second-order optimization methods, such as L-BFGS optimizer, to further decrease our loss function to improve the accuracy of our learned solutions.

\v \v \noindent {\bf Acknowledgement}

The work  was supported by the National Natural Science Foundation of China  under Grant No. 11925108.

\end{document}